%% file: congruence.tex
\documentclass{article}

\usepackage[a4paper,left=4.5pc,right=4.5pc,top=4.5pc,bottom=6pc]{geometry}

\fontsize{10pt}{10.2pt} \selectfont

\usepackage[utf8]{inputenc}

\usepackage{mathtools}

\usepackage{graphicx}
\usepackage{epstopdf}

\usepackage{subcaption}

\usepackage{cite}
\usepackage{url}

\usepackage{amsmath}
\usepackage{amssymb}
\usepackage{amsthm}

\usepackage{enumerate}

\newtheorem{theorem}{Theorem}[section]

\newtheorem{lemma}[theorem]{Lemma}
\newtheorem{proposition}[theorem]{Proposition}

\newtheorem{definition}[theorem]{Definition}
\newtheorem{example}[theorem]{Example}
\newtheorem{examples}[theorem]{Examples}

\newtheorem{claim}{Claim}

\usepackage{xcolor}
\definecolor{myOrange}{HTML}{E65C0D}
\definecolor{myBlue}{HTML}{5D6DC2}
\usepackage{algorithm}
\usepackage{listings}
\definecolor{colortodobg}{cmyk}{0,0,0.2,0}

\newcommand{\orth}{\ensuremath{\mathcal{MO}}}
\newcommand{\BoolMatrices}{\orth(k)_{\mathcal L}}
\newcount\colveccount
\newcommand*\colvec[1]{
        \global\colveccount#1
        \begin{pmatrix}
        \colvecnext
}
\def\colvecnext#1{
        #1
        \global\advance\colveccount-1
        \ifnum\colveccount>0
                \\
                \expandafter\colvecnext
        \else
                \end{pmatrix}
        \fi
}

\renewcommand{\geq}{\geqslant}
\renewcommand{\leq}{\leqslant}

\newcommand{\deff}{\coloneqq}

\newcommand{\NN}{\ensuremath{\mathbb{N}}}
\newcommand{\RR}{\ensuremath{\mathbb{R}}}
\newcommand{\RRposs}[1]{\ensuremath{\RR_{\geq #1}}}
\newcommand{\RRpos}{\RRposs{0}}

\newcommand{\BB}{\ensuremath{\mathcal{B}}}

\newcommand{\set}[1]{\ensuremath{\{#1\}}}
\newcommand{\setc}[2]{\set{#1 : #2}}

\newcommand{\PowerOfTwo}{\ensuremath{\textsl{Pow2}}}
\newcommand{\SP}[1]{\ensuremath{\langle #1 \rangle}}
\newcommand{\Norm}[1]{\ensuremath{\left\|#1\right\|}}

\newcommand{\TS}{\ensuremath{\mathcal{T}}} 
\newcommand{\trecvid}{\ensuremath{\mathcal{TR}}}
\newcommand{\MM}{\ensuremath{\mathbb{M}}} 

\newcommand{\dE}{\ensuremath{\textsf{\upshape d}}}

\renewcommand{\phi}{\ensuremath{\varphi}}
\renewcommand{\epsilon}{\ensuremath{\varepsilon}}

\newcommand{\nicht}{\ensuremath{\neg}}
\newcommand{\und}{\ensuremath{\wedge}}
\newcommand{\oder}{\ensuremath{\vee}}

\newcommand{\Und}{\ensuremath{\bigwedge}}
\newcommand{\Oder}{\ensuremath{\bigvee}}

\newcommand{\DoneCcomput}{\ensuremath{d_1^C\textsc{-Computation}}}
\newcommand{\Problem}[1]{\ensuremath{\textsc{#1}}}

\usepackage{authblk}

\begin{document}
\twocolumn

\title{Measuring Congruence on High Dimensional Time Series}

\author[1]{J\"org P. Bachmann}
\author[2]{Johann-Christoph Freytag}
\author[3]{Benjamin Hauskeller}
\author[4]{Nicole Schweikardt}
\affil[1,2,3,4]{ Humboldt-Universität zu Berlin, Germany\\ {\{bachmjoe,freytag,hauskelb,schweikn\}@informatik.hu-berlin.de}}
\date{\today}

\maketitle

\begin{abstract}
    A time series is a sequence of data items; typical examples are videos, stock ticker data, or streams of temperature measurements.
    Quite some research has been devoted to comparing and indexing simple time series, i.\,e., time series where the data items are real numbers or integers.
    However, for many application scenarios, the data items of a time series are not simple, but high-dimensional data points.
    E.\,g., in video streams each pixel can be considered as one dimension, leading to $k$-dimensional data items with $k=12,288$ already for low resolution videos with $128\times 96$ pixels per frame.

    Motivated by an application scenario dealing with motion gesture recognition, we develop a distance measure (which we call congruence distance) that serves as a model for the approximate congruency of two complex time series.
    This distance measure generalizes the classical notion of congruence from point sets to complex time series.

    We show that, given two input time series $S$ and $T$, computing the congruence distance of $S$ and $T$ is NP-hard.
    Afterwards, we present two algorithms with quadratic and quasi-linear runtime, respectively, that compute an approximation of the congruence distance.
    We provide theoretical bounds that relate these approximations with the exact congruence distance, as well as experimental results, which indicate that our approach yields accurate approximations of the congruence distance.
\end{abstract}



\input{introduction.tex}

\input{preliminaries.tex}

\input{timeseries.tex}

\input{tscongruence.tex}

\input{reducingcomplexity.tex}

\input{experiments.tex}

\input{conclusion.tex}

\input{appendix.tex}

\end{document}

%% file: introduction.tex
\section{Introduction}\label{sec:introduction}

Similarity search or nearest neighbour search is a common problem in computer science and has a wide range of applications (see Section~\ref{sec:relatedwork} for examples).
Given a dataset (in our case, a set of time series), a query (in our case, a time series), the problem is to find nearest neighbours to the query in the dataset, regarding a certain distance or similarity function.
The difference between distance and similarity functions is that a distance function returns $0$ for exact matches and a higher value otherwise, whereas similarity functions return greater values for more similar input data.
In this paper, we consider distance functions only.
There are two main variations of the nearest neighbour search problem.
The first variation is called the $\varepsilon$-nearest neighbour search ($\varepsilon$-NN search), where the search returns all elements from the dataset having a distance of at most $\varepsilon$ to the query.
The second variation is called Top-$k$ nearest neighbour search, where those $k$ elements having the smallest distance to the query will be returned.
In each case, a requirement in practical systems is the fast computation of the distance function.

The datasets considered in this paper are time series, i.e., sequences of elements in $\MM$, for
a metric space $(\MM,\dE)$. 
Examples of time series include simple time series where $\MM=\RR$ (e.\,g.\ temperature measurements or stock data) and multi variate time series where $\MM=\RR^k$ (e.\,g.\ motion trackings in three dimensional space or videos).

The distance functions defined and analyzed in this paper measure the (approximate) congruence of two time series.
Thereby, the distance between two time series $S$ and $T$ shall be $0$ iff two $S$ can be transformed into $T$ by rotation, translation, and mirroring; in this case, $S$ and $T$ are said to be \emph{congruent}.
A value greater than $0$ shall correlate to the amount of transformation needed to turn the time series into congruent ones.

\subsection{Motivation and Related Work}
\label{sec:relatedwork}

Simple time series are finite squences holding one number per time step.
There is a vast field of applications for simple time series in likely all scientific areas, including geo science (temperature measurements, earthquake prediction), medicine (heart rate measurements), and finance (stock ticker data).
Depending on the application, different similarity measurements of time series are used (e.\,g. Landmarks \cite{2000:LNM:846219.847378}, Dynamic Time Warping \cite{4400350}, and Longest Common Subsequence \cite{4400350}).
Different techniques evolved to speed up nearest neighbour searches \cite{Camerra:2010:IIM:1933307.1934553,Vlachos:2006:IMT:1146466.1146467}.
Esling and Agon published a survey on simple time series \cite{Esling:2012:TDM:2379776.2379788}.

Let us continue with a few examples highlighting the role of multi-dimensional time series.

\paragraph{Motion Gesture Recognition}

The interest in motion gesture recognition has drastically increased over the last decade, especially in combination with augmented reality systems, as for example the Oculus Rift \cite{Oculus}.
Recent products, like the LeapMotion \cite{LeapMotion} or Microsoft Kinect \cite{Kinect}, are able to recognize the posture of the hands and body, respectively.
These applications belong to appearence based approaches of motion gesture recognition, since they use cameras to recognize the posture at each time.
A second category of posture recognition systems include gloves \cite{Dipietro:2008:SGS:2220435.2221135}, which is more than 30 years old.
The area of their applications has grown more and more from medicine and health care up to recent applications as, for example, controlling a Smartphone \cite{Huber:2014:MTC:2540930.2540954,Gollner:2012:MLG:2148131.2148159}.
Approaches using systems like these gloves are called skeletal based.
The main difference is, that the gesture recognition software retrieves the key information, i.\,e. the trajectory of the body parts, instead of one or multiple video streams of that person.

Our interest, and the application of our work for motion gesture recognition, is the classification of gestures rather than the capturing itself.
There are various different approaches to treat this problem, e.\,g. Computer Vision based techniques \cite{Yang98extractionand}, trajectory based techniques \cite{iGesture}, approaches based on State Machines \cite{Hong00constructingfinite}, etc.

Considering the motion of a finger tip and its direction as a time series in $\mathbb R^6$, our approach contributes to the skeletal based algorithms.
From our point of view, the problem of motion gesture recognition narrows down to the problem of finding the most similar time series.
Hereby, similarity of two time series means the measurement of their congruence.
Since motion gestures usually are not performed exactly as stored in a database, we need a fine granular or approximative congruence measurement.
For example, a circle can be drawn more like an ellipse, but is more congruent to a circle than to a square or a line (see Figure~\ref{fig:motions}~and~Figure~\ref{fig:congruence}).
To the best of our knowledge, in the literature scaling and translation invariant, but no rotation invariant approaches have been developed.
However, the rotation invariance is a necessary feature for applications as for example interactive tables with multiple persons standing at all sides.

\begin{figure}
    \centering
    \includegraphics[width=.4\textwidth]{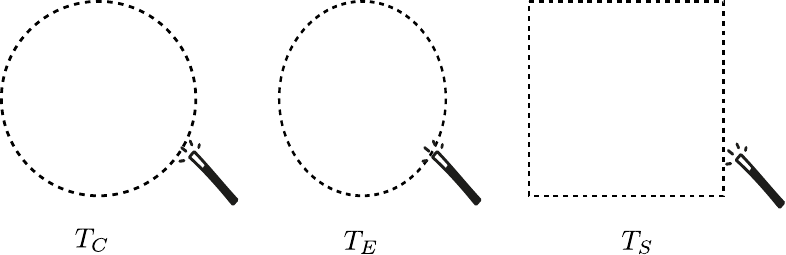}
    \caption{Three motions: A circle ($T_C$), an ellipse ($T_E$), and a square ($T_S$).}
    \label{fig:motions}
\end{figure}

\begin{figure}
    \centering
    \includegraphics[width=.3\textwidth]{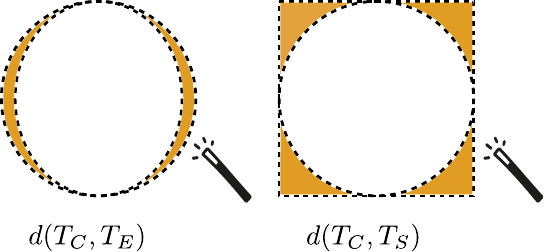}
    \caption{Sketches of distances: circle to ellipse (left), circle to square (right).}
    \label{fig:congruence}
\end{figure}

\paragraph{Content Based Video Copy Detection}

Nowadays, a vast amount of video data is uploaded and shared on community sites such as YouTube or Facebook.
This leads to various tasks such as copyright protection, duplicate detection, analysing statistics of particular broadcast advertisements, or searching for large videos containing certain scenes or clips.
Two basic approaches exist to address these challenges, namely watermarking and content based copy detection (CBCD).
Watermarking suffers from being vulnerable to transformations frequently performed during copy creation of a video (e.\,g.\ resizing or reencoding).
Furthermore, watermarking cannot be used on videos unmarked before distribution.
In contrast, CBCD is about finding copies of an original video by specifically comparing the contents and is thus more robust against transformations done during copy creation.
These transformations include resolution, format, and encoding changes, addition of noise, bluring, flipping, (color) negation, and gray-scaling.
Hence, copies are near-duplicates and it is natural to use a distance or similarity function to discover them.

Many approaches compare features created per image \cite{fastsequencematching,mpeg7}.
Global features include mean color values and color histograms.
In contrast to global features, local features (e.\,g. Harris Corners, SIFT, or SURF) are more robust against transformations when searching for similar images \cite{efficientnear-duplicate,zgridcbcd,scalablevideodbmining}.
However, these techniques suffer from weak robustness against transformations as for example flipping or negation.

Considering a video with $k$ pixels per image as a time series in a $k$ dimensional vector space, the transformations flipping, negation, and gray-scaling correspond to mirroring, rotating, and translating the time series and thus do not change the congruence distance to another video.
Furthermore, a global or local feature could be stored per image and regarded as state per time step.
Hence, the congruence distance function introduced in the present paper seems to be a good basis for video distance functions in combination with already existing techniques.

\paragraph{Congruence Calculation}

The classical \textsc{Congruence} problem basically determines whether two point sets $A,B\subseteq\mathbb R^k$ are congruent considering isometric transformations (i.\,e., rotation, translation, and mirroring) \cite{Heffernan:1992:ADA:142675.142697,Alt:1988:CSS:44611.44614}.
For two and three dimensional spaces, there are results
providing algorithms with runtime\linebreak[4] 
$\mathcal O(n\cdot\log n)$ \cite{Alt:1988:CSS:44611.44614}.
For larger dimensionalities, they provide an algorithm with runtime $\mathcal O(n^{k-2}\cdot\log n)$.
For various reasons (e.\,g. bounded floating point precision, physical measurement errors), the \emph{approximated} \textsc{Congruence} problem is of much more interest in practical applications.
Different variations of the approximated \textsc{Congruence} problem have been studied (e.\,g. what types of transformations are used, is the assignment of points from $A$ to $B$ known, what metric is used) \cite{Heffernan:1992:ADA:142675.142697,Alt:1988:CSS:44611.44614,Indyk:2003:ACN:636968.636974,Alt96discretegeometric}.

The \textsc{Congruence} problem is related to our work, since the problem is concerned with the existance of isometric functions such that a point set maps to another point set.
The main difference is, that we consider ordered lists of points (i.\,e. time series) rather than pure sets.


\subsection{Main Contributions}

In this paper, we use a model for complex time series covering models of time series known from the literature as well as high dimensional time series.
Focusing on high dimensional time series, our main contributions are as follows:

\begin{enumerate}[\ 1.]
    \item We define and analyze an intuitive \emph{congruence measurement} (congruence distance) which
  can be computed by solving an optimization problem with highly nonlinear constraints.
    \item We show that the calculation of the congruence distance is an NP-hard problem.
 This is done by constructing a technically involved polynomial time reduction from the NP-hard \textsc{1-in-3-Sat} problem.
    \item 
We provide two \emph{approximations} to the congruence distance (delta distance, and reduced delta distance) that
can be computed in polynomial time. Studying their approximativity, we obtain:
        \begin{itemize}
            \item The approximations yield \emph{lower bounds} on the congruence distance.
            \item There exist pathetic examples revealing that the \emph{relative error can grow arbitrarily}.
            \item Our \emph{experimental results} suggest a stable behaviour of the approximations in practical applications.
        \end{itemize}
\end{enumerate}

\subsection{Organization}

The rest of this paper is structured as follows.
In Section~\ref{sec:preliminaries} we provide basic notation used throughout the paper.
In Section~\ref{sec:timeseries}, we fix the notion of time series, and we present distance measures that 
turn the set of all time series into a metric space.
Section~\ref{sec:congruence} discusses the congruence of two time series:
We define an intuitive function measuring the congruence similarity of two time series and show that its calculation is an $\textup{NP}$-hard problem.
Furthermore, we provide an approximation with quadratical runtime and compare both distance functions with each other.
In Section~\ref{sec:approximatedcongruence} 
we provide an approximation which has quasi-linear runtime.
There are examples where the difference between the congruence distance functions provided in this paper grows arbitrarily.
However, the experimental results presented in Section~\ref{sec:experiments} indicate that in practice, our approach yields accurate approximations.
Section~\ref{sec:conclusion} concludes the paper.

%% file: preliminaries.tex
\section{Preliminaries}
\label{sec:preliminaries}

\paragraph{Basic notation}
By $\NN$, $\RR$, $\RRposs{c}$ we denote the set of non-negative integers, the set of reals, and the set of all reals $\geq c$, for some $c\in \RR$, respectively.
For integers $x,y$ we write $[x,y]$ for the interval consisting of all \emph{integers} $z$ with $x\leq z \leq y$, and we write $[x,y)$ for $[x,y]\setminus\set{y}$.

By $\RR^k$ and $\RR^{k\times k}$, for $k\in\NN$, we denote the set of all vectors of length $k$, resp.,  all $(k\times k)$-matrices with entries in $\RR$.
For a vector $v\in\RR^k$ we write $v_i$ for the entry in position $i$.

Similarly, for a matrix $M\in\RR^{k\times k}$ we write $m_{i,j}$ for the entry in row $i$ and column $j$.
By $e_i$ we denote the \emph{$i$-th unit vector} in $\RR^k$, i.e., the vector with entry $1$ in the $i$-th position and entry $0$ in all other positions.

We write $Mv$ for the product of the matrix $M\in\RR^{k\times k}$ and the vector $v\in\RR^k$.
We write $\lambda v$ and $\lambda M$ for the product of the number $\lambda\in\RR$ with the vector $v$ and the matrix $M$, respectively (i.e., for all $i,j\in[1,k]$, the $i$-th entry of $\lambda v$ is $\lambda v_i$, and the entry in row $i$ and column $j$ of $\lambda M$ is $\lambda m_{i,j}$).

By $\Norm{\cdot}_p$, for $p\in\RRposs{1}$, we denote the usual \emph{$p$-norm} on $\RR^k$; i.e., $\Norm{v}_p=\big({\sum_{i=1}^k |v_i|^p}\big)^{1/p}$ for all $v\in\RR^k$.

By $\SP{\cdot,\cdot}$ we denote the usual \emph{scalar product} on $\RR^k$; i.e., for $u,v\in\RR^k$ we have $\SP{u,v}=\sum_{i=1}^k u_i v_i$.
In particular, $\Norm{v}_2 = \sqrt{\SP{v,v}}$ for all $v\in\RR^k$.
Recall that two vectors $u,v\in\RR^k$ are orthogonal iff $\SP{u,v}=0$.

A matrix $M\in \RR^{k\times k}$ is called \emph{orthogonal} if the absolute value of its determinant is 1. Equivalently, $M$ is orthogonal iff $\SP{m_i,m_i}=1$ and $\SP{m_i,m_j}=0$ for all $i,j\in[1,k]$ with $i\neq j$, where $m_i$ denotes the vector in the $i$-th column of $M$.
We write $\orth(k)$ to denote the set of all orthogonal matrices in $\RR^{k\times k}$.
Recall that angles and lengths are invariant under multiplication with orthogonal matrices, i.\,e.:
\begin{align*}
    \forall\ u,v\in\mathbb R^k,\ M\in\orth(k)\ : \ \ & \langle Mu,Mv \rangle = \langle u, v \rangle. \\
    \forall\  u\in\mathbb R^k,\ M\in\orth(k)\ : \ \ & \|Mu\|_2 = \|u\|_2.
\end{align*}

In general, a \emph{vector norm} is an arbitrary mapping $\Norm{\cdot}:\RR^k\longrightarrow\RRpos$ that satisfies the following axioms:
\begin{align*}
    \forall\ v\in\mathbb R^k \ : \ \ & \|v\| = 0 \ \Longrightarrow \ v = 0. \\
    \forall\ \lambda\in\mathbb R,\ v\in\mathbb R^k \ : \ \ & \|\lambda v\| = |\lambda|\cdot\|v\|. \\
    \forall\ u,v\in\mathbb R^k \ : \ \ & \|u+v\| \leq \|u\| + \|v\|.
\end{align*}
Clearly, $\Norm{\cdot}_p$ is a vector norm ($\ell_p$ norm) for any $p\in\RRposs{1}$.

A \emph{matrix norm} is a mapping $\|\cdot\|:\mathbb R^{k\times k} \longrightarrow \RRpos$ satisfying the following axioms:
\begin{align*}
    \forall\ M\in\mathbb R^{k\times k} \ : \ \ & \|M\| = 0 \
    \Longrightarrow \ M=0. \\
    \forall\ \lambda\in\mathbb R,\ M\in\mathbb R^{k\times k} \ : \ \ & \|\lambda M\| = |\lambda|\cdot\|M\|. \\
    \forall\ M,M'\in\mathbb R^{k\times k} \ : \ \ & \|M+M'\|\leq \|M\|+\|M'\|.
\end{align*}
The particular matrix norms considered in this paper are the \emph{max column norm} $\|\cdot\|_m$ and the \emph{$p$-norm} $\Norm{\cdot}_p$, for $p\in\RRposs{1}$, which are defined as follows: For all $M\in\RR^{k\times k}$,
\begin{eqnarray*}
 \Norm{M}_m & := & \max_{j\in[1,k]}\Big(\sum_{i=1}^k |m_{i,j}|\Big), \\
 \Norm{M}_p & := & \Big( \sum_{i=1}^k\sum_{j=1}^k |m_{i,j}|^p\Big)^{1/p}.
\end{eqnarray*}

Recall that a \emph{pseudo metric space} $(\MM,d)$ consists of a set $\MM$ and a distance function $d:\MM\times \MM \longrightarrow \RRpos$ satisfying the following axioms:
\begin{align*}
    \forall\ x,y\in \MM \ : \ \ & d(x,y) = d(y,x). \\
    \forall\ x,y,z\in \MM \ : \ \ & d(x,z) \leq d(x,y)+d(y,z).
\end{align*}
A \emph{metric space} is a pseudo metric space which also satisfies
\begin{align*}
    \forall\ x,y\in \MM \ : \ \ & d(x,y) = 0 \ \Longleftrightarrow \ x = y.
\end{align*}
Note that if $\Norm{\cdot}$ is an arbitrary vector norm and $d(\cdot,\cdot)$ is defined as $d(u,v):= \Norm{u-v}$, then $(\RR^k,d)$ is a metric space.
By $\dE_p$, for $p\in\RRposs{1}$, we denote the usual $\ell_p$-distance, i.e., the particular distance function with $\dE_p(x,y)= \Norm{x-y}_p$.

If $\Norm{\cdot}$ is an arbitrary matrix norm and $d(\cdot,\cdot)$ is defined as $d(M,M'):=\Norm{M-M'}$ for all matrices $M,M'\in\RR^{k\times k}$, then $(\RR^{k\times k},d)$ is a metric space.


%% file: timeseries.tex
\section{Time Series}
\label{sec:timeseries}

Let $\MM$ be an arbitrary set. A \emph{time series over $\MM$} is a finite
sequence of elements in $\MM$. 
For a time series $T=(t_0,\ldots,t_{n-1})\in \MM^n$, we write $\#T$ to
denote the length $n$ of $T$. The elements $t_i$ of $T$ are called the
\emph{states} of $T$.

The special case where $\MM=\RR$ yields the \emph{simple time series}
that are usually considered in the literature; examples of application
areas are time sequences obtained from stock data, temperature
measurements or heart rate monitoring (here, we consider time series with
homogenous time intervals only).
For such simple time series, the \emph{distance} between two time
series $S$ and $T$ of equal length $n$ usually is 
defined as $\Norm{ \,|S-T|\,}$, where $\Norm{\cdot}$ is a vector norm and $|S-T|$ is
the vector in $\RR^k$ whose $i$-th entry is the real number $|s_{i-1}-t_{i-1}|$.
The most common case considered in the literature uses the 1-norm
$\|\cdot\|_1$, cf.\ e.g.\ \cite{Esling:2012:TDM:2379776.2379788}; see Figure~\ref{fig:simpletimeseries} for an illustration.

\begin{figure}
    \centering
    \includegraphics[width=.45\textwidth]{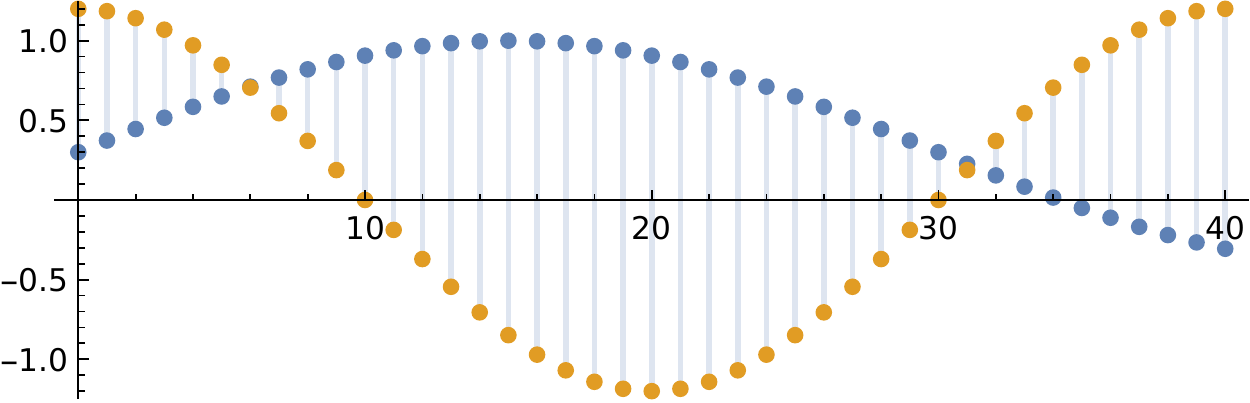}
    \caption{Two time series $S$ and $T$ in $\mathbb R$.}
    \label{fig:simpletimeseries}
\end{figure}

We generalize this to time series over arbitrary sets $\MM$ as follows.
Let $(\MM,\dE)$ be a metric space. For time series $S,T$ of length $n$
over $\MM$, we let 
$\dE(S,T)$ be the real vector of length $n$ with entry $\dE(s_i,t_i)$ in
its $(i{+}1)$-th position (for all $i\in[0,n)$).
Now let $\Norm{\cdot}$ be an arbitrary vector norm.
We define a distance measure $\Norm{\dE}: \MM^n \times \MM^n \longrightarrow
\RRpos$ via
\[
   \Norm{\dE}(S,T) \ := \ \  \Norm{\dE(S,T) }.
\]

By $\mathcal{T}_{\MM}$ we denote the set of all time series over
$\MM$ of arbitrary length, i.e., $\mathcal{T}_{\MM} =\bigcup_{n\in \NN} \MM^n$.
If $M$ is clear from the context, we will omit the 
subscript $\MM$ and simply write $\mathcal{T}$
instead of $\mathcal{T}_{\MM}$. For $n\in\NN$ we then write
$\mathcal{T}_n$ to denote the set $\MM^n$ of all time series of length
$n$ over $\MM$.
It is straightforward to verify the following.
\begin{proposition}
$(\mathcal{T}_n,\Norm{\dE}_p)$ is a metric space.
\end{proposition}

Next, we want to extend $\Norm{\dE}$ to a distance measure on
time series of arbitrary length, i.e., we want to extend $\Norm{\dE}$ 
to a mapping $\TS_{\MM}\times\TS_{\MM}\longrightarrow \RRpos$.
For this, the following notation is convenient.

\begin{definition}\upshape
 Let $T=(t_0,\cdots,t_{n-1})\in\mathcal T$ be a time series, 
 let $b\in[0,n)$, and let $\ell\in[1,n{-}b]$.
 Then $T_b^\ell \coloneqq (t_b,\cdots,t_{b+\ell-1})$ is the
 \emph{subseries} of $T$ of length $\ell$ starting at index $b$.
\end{definition}

If $S=(s_0,\ldots,s_{m-1})\in \MM^m$ and $T=(t_0,\ldots,t_{n-1})\in \MM^n$
are two time series of lenghts $m<n$, then we let
\[
  \Norm{\dE}(T,S) \ := \ \ 
  \Norm{\dE}(S,T) \ := \ \ 
  \min_{b\in [0,n-m]} \Norm{\dE}(S, T_b^m).
\] 
I.e., the distance between $S$ and $T$ is computed by finding the best
match of the shorter time series regarded as a window over the longer time series.
We will write 
\[
  d_p(\cdot,\cdot)
\] 
instead of
$\Norm{\dE}(\cdot,\cdot)$ for the special case where $\MM=\RR^k$, 
$\Norm{\cdot}=\Norm{\cdot}_p$ for some $p\in\RRposs{1}$, and 
$\dE$ is the \emph{Euclidean} distance $\dE_2$
defined via 
$\dE_2(x,y)=\Norm{x-y}_2$ for all $x,y\in\RR^k$.

It is easy to see that many other distance functions (e.\,g. DTW and
LCSS \cite{1104847,Corradini:2001:DTW:882476.883586,4400350}) 
that have been considered in the literature for time series over $\RR$ or $\RR^k$ can be
adopted to time series over $\MM$ for a metric space $(\MM,\dE)$ accordingly.

To avoid confusion between $\dE$,
$\Norm{\dE}$, $\dE_p$, $d_p$, and further distance functions considered in this
paper, we will henceforth write 
$d$ (or variants thereof) to denote distance functions for relating
time series (i.e., $d$ will be a function from $\mathcal{T}_\MM\times
\mathcal{T}_\MM$ to $\RRpos$),
and we will write 
$\dE$ (or variants thereof) to denote distance functions for relating
individual states in the time series (i.e., $\dE$ will be a function
from $\MM\times\MM$ to $\RRpos$).
The latter will be called \emph{state distance function}.

We will speak of \emph{metric time series} whenever considering time
series over $\MM$ for a metric space $(\MM,\dE)$. For a given vector norm
$\Norm{\cdot}$, the associated function $\Norm{\dE}$ will serve
as a distance measure for time series over $\MM$.

Let us conclude this section with a few examples that illustrate
the generality of metric time series.

\begin{examples}\upshape
 As already explained above, \emph{simple time series} are a
 special case of time series where
 $\MM=\mathbb R$, $\dE:\RR\times \RR\longrightarrow \RRpos$ is defined
 via $\dE(x,y)=|x-y|$ for $x,y\in\RR$, and
 $\Norm{\cdot}=\Norm{\cdot}_p$ for some $p\in\RRposs{1}$.

 \emph{Complex time series}, i.e., time series where the states are elements in $\RR^k$
 for some fixed $k$, are the special case where
 $\MM=\RR^k$, $\dE:\RR^k\times\RR^k\longrightarrow\RRpos$ is the
 Euclidean distance $\dE_2$,
 $\Norm{\cdot}=\Norm{\cdot}_p$ for some $p\in\RRposs{1}$, and hence
 $\Norm{\dE}=d_p$.

 For an arbitrary undirected connected graph $G=(V,E)$, we can consider the 
 mapping $\dE:V\times V\longrightarrow \RRpos$ where $\dE(u,v)$ is the
 length of a shortest path between nodes $u$ and $v$ of $G$. Note that
 $(V,\dE)$ is a metric space. Given an arbitrary vector norm
 $\Norm{\cdot}$, we can view sequences of nodes of $G$ as
 time series over $\MM=V$, and $\Norm{\dE}$ as a distance measure
 between such time series.
\end{examples}

In the remainder of this paper we restrict attention to time series over $\MM=\RR^k$ and state distance functions $\dE_p$.


%% file: tscongruence.tex
\section{Time Series Congruence}
\label{sec:congruence}

Let $\MM:=\RR^k$ and let $\TS:=\TS_{\MM}$. 
If $T=(t_0,\ldots,t_{n-1})\in\TS$ is a time series, $M\in\RR^{k\times
  k}$ is a matrix, and $v\in\RR^k$ is a vector, we write
$M\cdot T + v$ for the time series $(t'_0,\ldots,t'_{n-1})$ where 
$t'_i= Mt_i +v$ for each $i\in[0,n)$.

We say that two time series $S,T\in\mathcal{T}$ are \emph{congruent}, if $S$
can be transformed into $T$ by rotation, mirroring, or translation.
This is formalized in the following definition.

\begin{definition}\upshape
 Consider the metric space $(\RR^k,\dE)$ for $\dE\deff\dE_2$.
 Two time series $S$ and $T$ of the same length $n$ are called
 \emph{congruent} (for short: $S\cong_C T$) if there is a matrix
 $M\in\orth(k)$ and a vector $v\in\RR^k$ such that 
 $T=M\cdot S + v$.
\end{definition}

It is easy to see that for each $n\in\NN$, the congruence relation $\cong_C$ is an 
\emph{equivalence relation} on the class of all time series over $\RR^k$ of
length $n$.

According to the motivation provided in Section~\ref{sec:introduction},
we aim at a distance measure that regards two time series $S$ and $T$
as very similar if $T$ is obtained from $S$ via rotation,
mirroring, or translation, i.e., which satisfies the following
\emph{congruence requirement}.

\begin{definition}[Congruence Requirement]
    \label{def:congruencerequirement}\upshape\ \\
Let $k\in\NN$, let $\MM=\RR^k$, and let $\TS=\TS_{\MM}$.
A function $d:\mathcal T\times\mathcal T\longrightarrow\mathbb R_{\geq
  0}$ satisfies the \emph{congruence requirement} iff
for all time series $S,T\in\TS$ the following is true: 
    \begin{align*}
        d(S,T) = 0 
        \ \iff\  S\cong_C T.
    \end{align*}
\end{definition}

The following example highlights some intuition for the congruence
distance function that is provided in Definition~\ref{def:congruencedistance}.
\begin{example}\upshape
    \label{ex:intuitioncongruencedistance}
    Consider the time series
    \[ \textstyle
        S\coloneqq \left( \colvec{2}{-4}{0}, \colvec{2}{0}{0},
          \colvec{2}{1}{0} \right)\ , \qquad
        T\coloneqq \left( \colvec{2}{0}{3}, \colvec{2}{0}{0}, \colvec{2}{1}{0} \right) 
    \]
    Obviously, $d_1(S,T)=5$.
    Now, let us rotate $T$ by 90 degress counterclockwise, i.\,e., let
    us compute $M\cdot T$ for the matrix
    \[
        M \coloneqq \begin{pmatrix}
            0 & -1 \\
            1 & 0
        \end{pmatrix}. 
       \quad\text{Then,}\quad
        M\cdot T = \left( \colvec{2}{-3}{0}, \colvec{2}{0}{0}, \colvec{2}{0}{1} \right)
    \]
    and
    $d_1(S,M{\cdot} T)=1+\sqrt 2 < 5$.

    Thus, without rotation, we need to add a vector of Euclidean
    length $5$ to the first state of $T$ in order to
    transform $T$ into $S$.
    But after rotating $T$ by 90 degrees counterclockwise, we only
    need to add a vector of length $1$ to the first state and a vector
    of length $\sqrt 2$ to the third state of $M \cdot T$ to obtain
    the time series $S$.

    Adding vectors to certain states can be interpreted as investing
    energy to make both time series having the same structure, i.e.,
    being ``congruent''.
    Hence, the congruence distance defined below can be viewed as a
    measure for the minimum amount of energy needed to make both time
    series congruent.
\end{example}

\begin{definition}[Congruence Distance]
    \label{def:congruencedistance}\upshape
    Let $k\in\NN$, $\MM=\RR^k$, $\TS=\TS_{\MM}$, and $p\in\RRposs{1}$.
       The \emph{congruence distance} $d^C_p(S,T)$ between two time
       series $S,T\in\TS$ is defined via
    \begin{align*}
        d^C_p(S,T) \coloneqq \min_{M\in\orth(k),\, v\in\mathbb R^k} d_p\left(S,\,M\cdot T+v\right).
    \end{align*}
\end{definition}

Note that, although $\orth(k)$ and $\RR^k$ are infinite sets, it can
be shown that the ``min'' used
in the definition of $d^C_p(S,T)$ does exist, and that for given $S,T$ there are
$M\in\orth(k)$ and $v\in\RR^k$ such that $d^C_p(S,T)=d_p(S,\, M\cdot
T+v)$; a proof can be found in the appendix.

It is not difficult to see that the following holds for $\MM=\RR^k$
and $\TS_n=\MM^n$ for each $n\in\NN$:
\begin{proposition}
    \label{prop:congruencetriangle}
    $(\mathcal T_n, d_p^C)$ is a pseudo metric space.
\end{proposition}
The proof is given in the appendix.

Obviously, calculating $d^{C}_p(S,T)$ for arbitrary $S,T\in\mathcal T$
is a nonlinear optimization problem that can be solved using numeric
solvers.
However the problem is computationally difficult: As we show in the
next subsection, already the calculation of $d_1^C$ is $\textup{NP}$-hard.

\subsection{NP-Hardness}
\label{sec:nphard}

In this subsection we restrict attention to $d_1=\Norm{\dE_2}_1$ and
the according congruence distance $d_1^C$. 
Consider the following problem:
\begin{quote}
 $\DoneCcomput$
 \begin{description}
   \item[Input:] A number $k\in\NN$ and two time series $S$ and $T$ of
     equal length over $\RR^k$.
   \item[Task:] Compute (a suitable representation of) the number $d_1^C(S,T)$.
 \end{description}
\end{quote}

This subsection's main result is:
\begin{theorem}
    \label{thm:reduction}
    If $\textup{P}\neq\textup{NP}$, then \\
    $\DoneCcomput$ cannot
    be solved in polynomial time.
\end{theorem}
The remainder of Subsection~\ref{sec:nphard} is devoted to the proof
of Theorem~\ref{thm:reduction}, which constructs a reduction from the $\textup{NP}$-complete problem $\Problem{1-in-3-Sat}$.
Recall that $\Problem{1-in-3-Sat}$ is the problem where the input
consists of a propositional formula $\Phi$ in 3-cnf, i.e., in conjunctive normal form where
each clause is a disjunction of literals over 3 distinct variables.
The task is to decide whether there is an assignment $\alpha$ that
maps the variables occurring in $\Phi$ to the truth values $0$ or $1$,
such that in each disjunctive clause of $\Phi$ exactly one literal is
satisfied by $\alpha$; we will call such an assignment $\alpha$ a
\emph{1-in-3 model of $\Phi$}.

Our reduction from $\Problem{1-in-3-Sat}$ to $\DoneCcomput_\mu$ will
proceed as follows: 
A given 3-cnf formula $\Phi$ with $k$ variables
$V_1,\dots,V_k$, will be mapped to two time series $\bar S_\Phi$ and
$\bar T_\Phi$ over $\mathbb R^k$, 
which represent the formula $\Phi$ and its variables, respectively.
Our construction of $\bar S_\Phi$ and $\bar T_\Phi$ will ensure that
for a certain number $c(\Phi)$ the following is true: \
$d_1^C(\bar S_\Phi,\bar T_\Phi) = c(\Phi)$ $\iff$ there is a
1-in-3 model of $\Phi$.

The basic idea for our choice of $\bar S_\Phi$ and $\bar T_\Phi$ is
that each dimension of $\mathbb R^k$ represents one variable. 
An orthogonal matrix, mirroring the $i$-th dimension then will
correspond to negating the $i$-th variable $V_i$. 

To formulate the proof, the following notation will be convenient.
For a propositional formula $\Phi$ with $k$ variables, we write
$V_1,\ldots,V_k$ to denote the variables occurring in $\Phi$.
A \emph{literal} over a variable $V_i$ is a formula
$L_i\in\set{V_i,\nicht V_i}$.
A \emph{disjunctive (conjunctive) 3-clause} is a formula 
$\Psi_I = \Oder_{i\in I} L_i$ ($\Psi_I = \Und_{i\in I} L_i$) with 
$L_i\in\set{V_i,\nicht V_i}$,  $|I|=3$, and
$I\subseteq[1,k]$.
A \emph{3-cnf formula} is a formula
$\Phi = \Und_{j=1}^m \Psi_j$, where $m\geq 1$ and each $\Psi_j$ is a
disjunctive 3-clause.

Furthermore, we will use the following notation for concatenating time
series.
Let $\ell\geq 1$, and let $S_j=(s^j_0,\ldots,s^j_{n_j-1})$ be a time
series over $\RR^k$ for each $j\in[1,\ell]$. Then, by
\[
  S_1 \times \cdots \times S_\ell
\]
we denote the time series 
\[
  \big(\,
    s^1_0,\ldots,s^1_{n_1-1}, \
    \ldots, \ 
    s^\ell_0,\ldots,s^\ell_{n_\ell-1}
  \,\big).
\]
If $i_1 < \cdots < i_\ell$ is an increasing sequence of integers
and $S_{i_j}$ is a time series over $\RR^k$, for each $j\in[1,\ell]$,
then for $I:=\set{i_1,\ldots,i_\ell}$ we let
\[
   \bigotimes_{i\in I} S_i
   \ \ \deff \ \ 
   S_{i_1} \times \cdots \times S_{i_\ell}.
\]

\paragraph{From a 3-cnf formula $\Phi$ to time series $\bar S_\Phi$
  and $\bar T_\Phi$}

For a given 3-cnf formula $\Phi$ let $k$ be the number of variables
occurring in $\Phi$.
Let $\Phi$ be of the form $\bigwedge_{j=1}^m \Psi_j$, where $m\geq 1$
and each $\Psi_j$ is a disjunctive 3-clause of the form $\Oder_{i\in
  I_j} L_{i}$, where $L_i\in\set{V_i,\nicht V_i}$, $|I_j|=3$, and  $I_j\subseteq [1,k]$.

For a disjunctive 3-clause $\Psi=\Oder_{i\in I}L_i$ let 
\[
  \Psi' \ \ \deff \ \
  \Oder_{j\in I}\ \Big(\
     \underbrace{L_{j} \ \und \!\!\! \Und_{i\in I\setminus\set{j}}
       \!\!\!\overline{L_{i}}}_{\textstyle =: \ \Gamma_{j}}
  \ \Big),
\]
where \ $\overline{V_{i}}\deff \nicht V_{i}$ \ and \
$\overline{\nicht V_{i}} \deff V_{i}$.
Clearly, an assignment $\alpha$ satisfies $\Psi'$ iff it
is a 1-in-3 model of $\Psi$. And $\alpha$ satisfies
\ $\Phi'\deff \Und_{j=1}^m\Psi'_j$ \
iff it is a 1-in-3 model of
\ $\Phi=\Und_{j=1}^m\Psi_j$.

The formulas $\Gamma_{i'}$ for $i'\in I$ are called \emph{the
  conjunctive 3-clauses implicit in $\Psi$}.

We define an embedding $\theta$ of variables, literals, and
conjunctive 3-clauses into $\RR^k$ as follows:
For each $i\in[1,k]$ let
\[
   \theta(V_i):= e_i
   \quad\text{and}\quad
   \theta(\nicht V_i):= -e_i.
\]
For a literal $L_i$ we let $l_i:=\theta(L_i)$.
For a conjunctive 3-clause $\Gamma=\bigwedge_{i\in I} L_i$, we let
\[
  \gamma \ \ \deff \ \ \theta(\Gamma) \ \ \deff \ \ \sum_{i\in
    I}\theta(L_i)
  \ \ = \ \ \sum_{i\in I}l_i.
\]
In particular, for $\Gamma_{j}$ as defined above, we obtain that
\[
   \gamma_{j} \ \ \deff \ \ \theta(\Gamma_{j}) \ \ = \ \ l_{j} -
   \sum_{i\in I\setminus\set{j}} l_i.
\]
For each disjunctive 3-clause \ $\Psi=\Oder_{i\in I}L_i$ \
we let \ 
\[\textstyle
   e_I \ \deff \ \sum_{i\in I}e_i
\]
and define the following time series over $\RR^k$:
\begin{equation}\label{eq:den:SPsi}
 \begin{array}{rlcrl}
     S'_\Psi  \coloneqq & \displaystyle\bigotimes_{i\in I}\ (6e_i, -6e_i)
   & , \
   & T'_\Psi  \coloneqq & \displaystyle\bigotimes_{i\in I}\ (6 e_i, 6 e_i),
 \medskip\\
     S_\Psi  \coloneqq & \displaystyle\bigotimes_{j\in I}\ (\gamma_j)
   & , \
   & T_\Psi  \coloneqq & \displaystyle (\, e_I,\; e_I,\; e_I\,)
 \medskip\\
     \tilde{S}_\Psi \coloneqq & S'_\Psi\times S_\Psi
   & , \ 
   &  \tilde{T}_\Psi \coloneqq & T'_\Psi\times T_\Psi.
 \end{array}
\end{equation}
For a 3-cnf formula $\Phi = \Und_{j=1}^m\Psi_j$ all these time series will be
concatened to the two time series
\[
   S_\Phi  \coloneqq \ \bigotimes_{j=1}^m \left(
   \tilde{S}_{\Psi_j}
  \right) 
   \quad , \quad
   T_\Phi  \coloneqq\ \bigotimes_{j=1}^m \left(
   \tilde{T}_{\Psi_j}
  \right).
\]
Finally, to be able to handle translations, we concatenate the time series with their mirrored duplicates:
\[
   \bar S_\Phi \coloneqq \ S_\Phi\times -S_\Phi 
   \quad \ \  , \quad \ \ 
   \bar T_\Phi \coloneqq \ T_\Phi\times -T_\Phi.
\]
Our aim is to compute a number $c(\Phi)$ such that the following is
true: 
$d_1^C(\bar S_\Phi, \bar T_\Phi)=c(\Phi)$ iff $\Phi$ has a 1-in-3
model.
For obtaining this, we will proceed in several steps, the first of
which is to compute a number $c^\orth(\Phi)$ such that $\Phi$ has a
1-in-3 model iff $d_1^{\orth}(S_\Phi,T_\Phi)=c^\orth(\Phi)$, for
\begin{equation}\label{eq:d1orth}
  d_1^\orth(S,T) \ \deff \ \ 
  \min_{M\in\orth(k)} d_1(S,\ M\cdot T).
\end{equation}

The idea behind our choice of the time series $S_\Phi$ and
$T_\Phi$ is as follows: $S'_\Psi$ and $T'_\Psi$ force the orthogonal
matrix $M$ to have a suitable shape when leading to the minimal
distance, i.\,e. to have all the $e_i$'s as Eigenvectors with
Eigenvalues of $1$ or $-1$ --- in other words: each vector
$e_i$ will either be left untouched or will be negated.
The time series $S_\Psi$ represents the disjunctive 3-clause $\Psi$, while $T_\Psi$ holds the vector representing the variables used in $\Psi$.
The minimum of $\tilde{S}_\Psi = S'_\Psi\times S_\Psi$ to
$\tilde{T}_\Psi= T'_\Psi\times T_\Psi$ will then be reached if the
vector $\sum_{i\in I}e_i$ is rotated in such a way that it matches one of the vectors of $S_\Psi$.
Hence, assigning a propositional variable $V_i$ the value 0
corresponds to negating the $i$-th dimension, and assigning $V_i$ the
value 1 leaves that dimension untouched.

\paragraph{Relating $d_1^\orth(S_\Phi,T_\Phi)$ with  1-in-3 models of
  $\Phi$}

The next observation will be helpful for our proofs.

\begin{lemma}\label{lem:Thales}
 Let $M\in\orth(k)$, \ $i\in[1,k]$, \ $a_i\deff
 \dE_2(e_i,Me_i)$, and $b_i\deff\dE_2(-e_i,Me_i)$.
 Then, 
 \[ \textstyle
   b_i= \sqrt{4-a_i^2}
   \qquad\text{and}\qquad
   a_i=\sqrt{4-b_i^2}.
  \]
\end{lemma}
\begin{proof}
 We let $c_i\deff\dE_2(e_i,-e_i)$.
 For the special case where $k=2$, Thales' Theorem tells us that
 $a_i^2+b_i^2=c_i^2$. The same holds true for arbitrary $k$, as the following
 computation shows.

 Clearly, $c_i=\dE_2(e_i,-e_i)=\Norm{2e_i}_2=2$, and thus $c_i^2=4$.
Furthermore, $a_i^2= \dE_2(e_i,Me_i)^2=\Norm{e_i-Me_i}_2^2 =\\
\SP{e_i-Me_i,e_i-Me_i} = \SP{e_i,e_i} + \SP{Me_i,Me_i}
-2\SP{e_i,Me_i} = 2 - 2\SP{e_i,Me_i}$.
And $b_i^2= \dE_2(-e_i,Me_i)^2 = \Norm{-e_i-Me_i}_2^2 =
\Norm{e_i+Me_i}_2^2=
\SP{e_i+Me_i,e_i+Me_i} = \SP{e_i,e_i} + \SP{Me_i,Me_i} +
2\SP{e_i,Me_i} = 2 + 2\SP{e_i,Me_i}$. 
Thus,
\[
  a_i^2 + b_i^2 
  \ = \  
  2- 2\SP{e_i,Me_i} + 2 + 2\SP{e_i,Me_i} 
  \ = \ 
  4
  \ = \ 
  c_i^2.
\]
Thus,
\ $b_i= \sqrt{4-a_i^2}$ \ and \ $a_i=\sqrt{4-b_i^2}$.
\end{proof}

From now on, whenever given a matrix $M\in\orth(k)$, we will always use the following notation:
\ $a_i\deff \dE_2(e_i,Me_i)$, \ and \ $b_i\deff \dE_2(-e_i,Me_i)$.
From Lemma~\ref{lem:Thales} we know that $b_i=\sqrt{4-a_i^2}$ and $a_i=\sqrt{4-b_i^2}$.
\medskip

For a disjunctive 3-clause $\psi$ and a matrix $M\in\orth(k)$ we let
\[
 \begin{array}{rrl}
  d_\psi(M) & \deff & d_1(\tilde{S}_\Psi, \ M\cdot\tilde{T}_\Psi) 
  \smallskip\\
  & = & d_1(S'_\Psi,\ M\cdot T'_\Psi) \ + \ d_1(S_\Psi,\ M\cdot T_\Psi).
 \end{array}
\]

In the next lemmas, we will gather information on the size of $d_\psi(M)$ (cf. the appendix for proof of Lemma~\ref{lem:d'} and Lemma~\ref{lem:Neu}).

\begin{lemma}\label{lem:d'}
    Let $\Psi=\bigvee_{i\in I} L_i$ be a disjunctive 3-clause, let
    $M\in\orth(k)$.
    Then,
    \begin{align}
        \label{eq:lem:d'}
        d_1(S'_\Psi,\ M\cdot T'_\Psi) \ = \ \ 6\cdot\sum_{i\in I}
        \big( a_i + b_i \big).
    \end{align}
\end{lemma}

\begin{lemma}\label{lem:d}
Let $\Psi=\bigvee_{i\in I} L_i$ be a disjunctive 3-clause.
\begin{enumerate}[(a)]
\item
For each $M\in\orth(k)$ we have
\begin{equation}\label{eq:lem:d_1}
  d_1(S_\Psi,M\cdot T_\Psi) \ \ \geq \ \
  4\sqrt{2} \ - \ 3\cdot\sum_{i\in I}\min(a_i,b_i).
\end{equation}
\item
Let $M$ be an element in $\orth(k)$ such that $Me_I=\gamma$, where $\gamma=\theta(\Gamma)$ for some
conjunctive 3-clause $\Gamma$ implicit in $\Psi$. Then
\ $d_1(S_\Psi,\ M\cdot T_\Psi) 
   =
   4\sqrt{2}$.
\item
Let $M$ be an element in $\orth(k)$ such that
$Me_i\in\set{e_i,-e_i}$ for all $i\in I$, and
$d_1(S_\Psi,M\cdot T_\Psi)=4\sqrt{2}$. 
Then $Me_I=\gamma$,
where $\gamma=\theta(\Gamma)$ for some conjunctive 3-clause $\Gamma$
implicit in $\Psi$.
\end{enumerate}
\end{lemma}
\begin{proof}
Let $\Gamma_j$, for $j\in I$, be the conjunctive 3-clauses implicit in
$\Psi$, and let $\gamma_j=\theta(\Gamma_j)$.

For proving \emph{(a)},
let $M$ be an arbitrary element in $\orth(k)$.
Note that by definition of $S_\Psi$ and $T_\Psi$ we have 
\begin{equation}\label{eq:lem:d_0}
  d_1(S_\Psi, \ M\cdot T_\Psi) \ \ = \ \
  \sum_{j\in I} \dE_2(\gamma_j,\ M e_I).
\end{equation}
By the triangle inequality and the symmetry
we know that $d(x,z)\geq d(x,y)-d(y,z)$ is true for all pseudo metric
spaces $(\MM,d)$ and all 
$x,y,z\in \MM$.
Thus, for any vector $v\in\RR^k$ and for any $j\in I$ we have
\[
   \dE_2(\gamma_j,Me_I) \ \ \geq \ \ 
   \dE_2(\gamma_j,v) \ - \ \dE_2(v,Me_I),
\]
and hence \  $d_1(S_\Psi, M\cdot T_\Psi) \geq $
\begin{equation}\label{eq:lem:d_3}
 \begin{array}{ll}
 & 
  \sum_{j\in I} \Big( \dE_2(\gamma_j,\ v) \ - \ \dE_2(v, Me_I)\,\Big)
 \smallskip\\
  = & 
  \Big(\sum_{j\in I}\dE_2(\gamma_j,v)\Big) \ \ - \ \ 3\cdot \dE_2(v,
  Me_I).
 \end{array}
\end{equation}
Let us choose $v\in\RR^k$ as follows: We let $v\deff \sum_{i\in I}s_i
e_i$ where $s_i\deff 1$ if $a_i\leq b_i$, and $s_i\deff -1$ otherwise.
Then, \ 
$\dE_2(v,Me_I) 
 = \Norm{v-Me_I}_2
 = \Norm{\sum_{i\in I} (s_i e_i - Me_i) }_2
 \leq$
\[
 \sum_{i\in I}\Norm{s_ie_i - Me_i}_2
 \ \ = \ \ \sum_{i\in I} \dE_2(s_ie_i, Me_i).
\]
Note that \ $\dE_2(s_ie_i,Me_i)$ is equal to $a_i$ if $s_i=1$, and it
is equal to $b_i$ if $s_i=-1$. Thus, due to our choice of $s_i$, we
know that $\dE_2(s_ie_i,Me_i) = \min(a_i,b_i)$, and hence
\begin{equation}\label{eq:lem:d_4}
  3 \cdot \dE_2(v,Me_I) \ \ \leq \ \ 3\cdot \sum_{i\in I}\min(a_i,b_i).
\end{equation}
Our next goal is to show that $\sum_{j\in I}\dE_2(\gamma_j,v)\geq
4\cdot\sqrt{2}$.
For simplicity let us consider w.l.o.g.\ the case where $I=\set{1,2,3}$.
For $i\in I$ let $l_i=\theta(L_i)$ (thus,
$l_i\in\set{e_i,-e_i}$). 
Then, w.l.o.g.\ we have
\[
\begin{array}{c}
 \gamma_1 = l_1 -l_2 -l_3, \qquad\quad
 \gamma_2 = -l_1 + l_2 - l_3,
\smallskip\\
 \gamma_3 = -l_1 - l_3 + l_3.
\end{array}
\]
For showing that \ $\sum_{j\in I}\dE_2(\gamma_j,v)\geq
4\cdot\sqrt{2}$, \ we make a case distinction according to $v$.
\smallskip

\emph{Case 1: $v=\gamma_i$ for some $i\in I$.} \ 
In this case, $\dE_2(\gamma_i,v) = 0$, and for each $j\in
I\setminus\set{i}$, it is straightforward to see that
$\dE_2(\gamma_j,v)=\sqrt{8}= 2\sqrt{2}$.
Thus, $\sum_{j\in I}\dE_2(\gamma_j,v)=4\sqrt{2} \ \approx \ 5.656$.
\smallskip

\emph{Case 2: $v=-\gamma_i$ for some $i\in I$.} \
In this case, $\dE_2(\gamma_i,v)=2\cdot\Norm{\gamma_i}_2 = 2\sqrt{3}$.
Furthermore, for each $j\in I\setminus\set{i}$, it is straightforward
to see that $\dE_2(\gamma_j,v)=\sqrt{4}=2$.
Thus, $\sum_{j\in I}\dE_2(\gamma_j,v)=4+2\sqrt{3} \ > \ 4\sqrt{2}$. 
\smallskip

\emph{Case 3: $v=l_1+l_2+l_3$.} \ 
Then, for each $j\in I$ we have
$\dE_2(\gamma_j,v)=\sqrt{4+4}=\sqrt{8}=2\sqrt{2}$.
Thus, $\sum_{j\in I}\dE_2(\gamma_j,v)=3\cdot 2\sqrt{2}= 6\sqrt{2} \ >
\ 4\sqrt{2}$.
\smallskip

\emph{Case 4: $v=-l_1-l_2-l_3$.} \
Then, for each $j\in I$ we have
$\dE_2(\gamma_j,v)=\sqrt{4}=2$. 
Thus, $\sum_{j\in I}\dE_2(\gamma_j,v)=3\cdot 2 = 6 \ > \ 4\sqrt{2}$.
\smallskip

Note that Cases~1--4 comprise all possible cases for $v$, and in all
these cases, $\sum_{j\in I}\dE_2(\gamma_j,v) \geq 4\sqrt{2}$.
Together with \eqref{eq:lem:d_3} and \eqref{eq:lem:d_4} we
obtain that equation \eqref{eq:lem:d_1} is correct. This completes the
proof of \emph{(a)}.
\medskip

For the proof of \emph{(b)}, let $M$ be an element in $\orth(k)$ such that
$Me_I=\gamma_i$, for some
$i\in I$.  From equation \eqref{eq:lem:d_0} we then obtain 
\[
  d_1(S_\Psi,M\cdot T_\Psi) 
  \ \ = \ \ \sum_{j\in I}\dE_2(\gamma_j,\gamma_i).
\]
According to Case~1 above, $\sum_{j\in
  I}\dE_2(\gamma_j,\gamma_i)=4\sqrt{2}$.
This completes the proof of \emph{(b)}.
\medskip

For the proof of \emph{(c)}, let 
$M$ be an element in $\orth(k)$ such that
\ $d_1(S_\Psi,M\cdot T_\Psi) = 4\sqrt{2}$ \ and 
$Me_i\in\set{e_i,-e_i}$ for all $i\in I$.
Thus, $Me_I$ is equal to a vector $v'=\sum_{i\in I}s'_ie_i$ where
$s'_i\in\set{1,-1}$ for each $i\in I$.
The above case distinction (for $v'$ rather than $v$) tells us that
\ $\sum_{j\in I}\dE_2(\gamma_j,v') = 4\sqrt{2}$ \ iff
\ $v'=\gamma_i$ for some $i\in I$.
This completes the proof of \emph{(c)}, since $\sum_{j\in
  I}\dE_2(\gamma_j,v') = d_1(S_\Psi,M\cdot T_\Psi)$.
\end{proof}

From the two previous lemmas, we easily obtain the following Lemma (cf. the appendix for a proof):

\begin{lemma}\label{lem:Neu}
Let $\Psi=\bigvee_{i\in I} L_i$ be a disjunctive 3-clause,
let $\Gamma$ be one of the conjunctive 3-clauses implicit in $\Psi$, and
let $\gamma\deff\theta(\Gamma)$.
For each $i\in I$ let $l_i\in\set{e_i,-e_i}$ be such that
$\gamma=\sum_{i\in I}l_i$.
Then, there exists an $M\in\orth(k)$ with $Me_i=l_i$ for each $i\in
I$. 

And for each $M\in\orth(k)$ satisfying $Me_i=l_i$ for all $i\in I$, we have
\begin{align*}
  d_\Psi(M) \quad = \quad 36 \ + \ 4\cdot\sqrt{2} \quad \approx \quad 41.656.
\end{align*}
\end{lemma}

\begin{lemma}\label{lem:dPsi}
Let $\Psi=\bigvee_{i\in I} L_i$ be a disjunctive 3-clause.
Then,
\[
  \min_{M\in\orth(k)} d_\Psi(M) \ \  = \ \ 36 \ + \ 4\cdot\sqrt{2}.
\]
Furthermore, every $M\in\orth(k)$ with $d_\Psi(M)=36+4\sqrt{2}$ has the
following properties:
\begin{enumerate}[($**$)]
 \item[($*$)]
   $Me_i \ \in \ \set{e_i,-e_i}$, \ \ for every $i\in I$.
 \item[($**$)]
   $Me_I = \gamma$, \ 
   where $\gamma=\theta(\Gamma)$ for a 
   conjunctive 3-clause $\Gamma$ implicit in $\Psi$.
\end{enumerate}
\end{lemma}
\begin{proof}
We first show that $d_\Psi(M)\geq 36+4\sqrt{2}$ is true for all $M\in \orth(k)$.
To this end, let $M$ be an arbitrary matrix in $\orth(k)$.
By the Lemmas~\ref{lem:d'} and \ref{lem:d} we know that
\ $d_\Psi(M) \ = \ d_1(S'_\Psi,M\cdot T'_\Psi) + d_1(S_\Psi,M\cdot
T_\Psi) \ \geq$
\begin{equation*}
 \begin{array}{ll}
 & \displaystyle
  6\cdot \sum_{i\in I} (a_i+b_i) \ \ + \ \ 4\sqrt{2} \ \ - \ \ 3\cdot
  \sum_{i\in I}\min(a_i,b_i)
 \smallskip\\
 = & \displaystyle
 4\sqrt{2} \ \ + \ \ \sum_{i\in I} \; \big(\;
   6 a_i \ + \ 6 b_i \ - \ 3\min(a_i,b_i)
 \; \big)
 \smallskip\\
 = & \displaystyle
 4\sqrt{2} \ \ + \ \ 3\cdot \sum_{i\in I} \; \underbrace{\big(\;
   2a_i \ + \ 2b_i \ - \ \min(a_i,b_i) 
 \;\big)}_{=:\ N_i}
 \end{array}
\end{equation*}
What is the smallest value possible for $N_i$?
Recall that $a_i=\sqrt{4-b_i^2}$ and $b_i=\sqrt{4-a_i^2}$.
Thus, in case that $a_i\leq b_i$ we have
\[
  N_i \ \ = \ \ a_i \ + \ 2\cdot\sqrt{4-a_i^2},
\]
and in case that $a_i>b_i$ we have
\[
  N_i \ \ = \ \ b_i \ + \ 2\cdot\sqrt{4-b_i^2}.
\]
Furthermore, if $a_i\leq\sqrt{2}$, then $b_i\geq \sqrt{2}$, and hence
$a_i\leq b_i$. If $a_i>\sqrt{2}$, then $b_i<\sqrt{2}$, and hence
$b_i<a_i$.
Therefore, $N_i$ is the minium value of the function 
\[
   f:\setc{x\in \RR}{0\leq x\leq \sqrt{2}} \longrightarrow \RR
\]
defined via
\[
   f(x) \ \deff \ \ x \ + \ 2\cdot \sqrt{4-x^2}.
\]
It is not difficult to verify that $f(0)=4$, and $f(x)>4$ for all $x\in\RR$ with $0 < x\leq \sqrt{2}$.

Thus, $N_i\geq 4$ is true for every $i\in I$. 
This leads to
\[
  d_\Psi(M) 
  \ \ \geq \ \
  4\sqrt{2} \ + \ 3\cdot\sum_{i\in I} 4 
  \ \ = \ \ 
  4\sqrt{2} \ + \ 3\cdot 3\cdot 4
  \ \ = \ \ 
  4\sqrt{2} \ + \ 36.
\]
Combining this with Lemma~\ref{lem:Neu} we obtain that
\[
  \min_{M\in\orth(k)} d_\Psi(M) \ \ = \ \ 36 \ + \ 4\sqrt{2}.
\]

Now let us consider an
arbitrary $M\in\orth(k)$ for which
$d_\Psi(M)=36+4\cdot\sqrt{2}$.
From the computations above we know that for each $i\in I$ it must be
true that $N_i=4$ and hence $a_i=0$ or $b_i=0$.
Since $a_i=\dE_2(e_i,Me_i)$ and $b_i=\dE_2(-e_i,Me_i)$, this implies
that $Me_i \in \set{e_i,-e_i}$.
Hence, the lemma's statement \emph{($*$)} holds.

Furthermore, since $a_i=\sqrt{4-b_i^2}$, we know that $\set{a_i,b_i}=\set{0,2}$.
Hence, Lemma~\ref{lem:d'} implies that
$d_1(S'_\Psi,M\cdot T'_\Psi)= 6\cdot 3\cdot 2= 36$.
By definition, we have 
\[
  d_\Psi(M) \ \ = \ \ d_1(S'_\Psi,M\cdot T'_\Psi) + d_1(S_\Psi,M\cdot
  T_\Psi);
\]
and by assumption we have \ $d_\Psi(M) = 36 + 4\sqrt{2}$.
Thus, \ $d_1(S_\Psi,M\cdot T_\Psi)=4\sqrt{2}$.
From Lemma~\ref{lem:d}\emph{(c)} we therefore obtain that
the lemma's statement \emph{($**$)} is correct.
\end{proof}

The previous lemma tells us, in particular, that each $M\in\orth(k)$, for which 
$d_\Psi(M)$ is minimal, belongs to the set
    \begin{align*}
        \BoolMatrices \ \coloneqq \ \left\{\ M\in\orth(k) \ : \
          m_i \in \set{e_i,-e_i} \ \right\}
    \end{align*}
  where $m_i$ denotes the vector in the $i$-th column of $M$.
  Henceforth, the elements in $\BoolMatrices$ will be called \emph{boolean matrices}.

  For each boolean matrix $M\in\BoolMatrices$ we let $A(M)$ be the assignment $\alpha$
  with $\alpha(V_i)=1$ if \;$m_{i,i}=1$, \;and $\alpha(V_i)=0$ if \;$m_{i,i}=-1$.
  Obviously, $A$ 
    is a bijection between $\BoolMatrices$ and
    the set of all assignments to the propositional variables
    $V_1,\dots,V_k$.

\begin{lemma}
    \label{lem:dphi}
    Let $\Phi=\Oder_{j=1}^m \Psi_m$  be a 3-cnf formula with $m$ disjunctive clauses.
    Then,
    $\Phi$ has a 1-in-3 model iff
    \[
       \min_{M\in\orth(k)}  d_1(S_\Phi,\ M\cdot T_\Phi) 
       \ \ = \ \ 
       m\cdot( 36+4\sqrt 2).
    \]
\end{lemma}
\begin{proof}
According to our definition of $S_\Phi$ and $T_\Phi$, the following is
true for every $M\in\orth(k)$: \ 
$d_1(S_\Phi,M\cdot T_\Phi) \ = \ \sum_{j=1}^m
d_1(\tilde{S}_{\Psi_j},M\cdot \tilde{T}_{\Psi_j})$.
Furthermore, by Lemma~\ref{lem:dPsi} we know for each $j\in[1,m]$ that 
\ $
   d_1(\tilde{S}_{\Psi_j},\ M\cdot \tilde{T}_{\Psi_j}) 
    \geq 
   36 \ + \ 4\sqrt{2}.
$ \ 
Thus, 
\[
   \min_{M\in\orth(k)}  d_1(S_\Phi,\ M\cdot T_\Phi) 
   \ \ \geq \ \ 
   m\cdot( 36+4\sqrt 2).
\]
To prove the lemma's ``if direction'', assume that there is an
$M\in\orth(k)$ such that 
\ $d_1(S_\Phi,\ M\cdot T_\Phi)  =
   m\cdot( 36+4\sqrt 2).
$

Then, for each $j\in[1,m]$ we have 
\[
    d_{\Psi_j}(M) \ \ = \ \ 
   d_1(\tilde{S}_{\Psi_j},\ M\cdot \tilde{T}_{\Psi_j}) 
    \ \ = \ \ 
   36 \ + \ 4\sqrt{2}.
\]
Thus, according to Lemma~\ref{lem:dPsi}, $M$ has the properties ($*$)
and ($**$). 
In particular, $M$ is a \emph{boolean matrix} in $\BoolMatrices$.
Let $\alpha:=A(M)$ be the variable assignment associated with $M$.
In the following, we show that $\alpha$ is a 1-in-3 model of $\Psi_j$, for each
$j\in[1,m]$. 

Fix an arbitrary $j\in[1,m]$ and let $\Psi:=\Psi_j$.
Let $I=\set{i_1,i_3,i_3}\subseteq [1,k]$ such that $\Psi= (L_{i_1} \oder L_{i_2} \oder L_{i_3})$, where $L_{i}$ is a literal over the variable $V_{i}$, for each $i\in I$.

From ($*$) and ($**$) we know that $Me_i\in\set{e_i,-e_i}$ for every $i\in I$, and $Me_{I}=\gamma$, where $\gamma=\theta(\Gamma)$ for a conjunctive 3-clause $\Gamma$ implicit in $\Psi$.
W.l.o.g., $\Gamma= (L_{i_1} \und \nicht L_{i_2} \und \nicht L_{i_3})$.
Thus, $\gamma=l_{i_1}-l_{i_2}-l_{i_3}$, where $l_i=\theta(L_i)$ for each $i\in I$.

For each $i\in I$ let $m_{i}$ be the vector in the $i$-th column of
$M$. Then, the following is true:
\[
  m_{i_1} + m_{i_2} + m_{i_3}
  \ \ = \ \ 
  Me_I
  \ \ = \ \ 
  \gamma
  \ \ = \ \ 
  l_{i_1} - l_{i_2} - l_{i_3}.
\]
Hence, \ $m_{i_1}=l_{i_1}$, \ $m_{i_2}=-l_{i_2}$, \ and \
$m_{i_3}=-l_{i_3}$.
Therefore, the associated variable assignment $\alpha:=A(M)$ satisfies
the literal $L_{i_1}$, but not the literals $L_{i_2}, L_{i_3}$.
Hence, $\alpha$ is a 1-in-3 model of $\Psi$.

In summary, we have shown that $\alpha$ is a 1-in-3 model of $\Psi_j$,
for each $j\in[1,m]$. Therefore, $\alpha$ also is a 1-in-3 model of
$\Phi$. This completes the proof of the ``if direction''.
\medskip

For the proof of the ``only-if direction'', let us consider the case
where $\Phi$ has a 1-in-3 model. I.\,e., there exists a variable
assignment $\alpha$ which, for each $j\in[1,m]$, satisfies exacly one
literal in the disjunctive clause $\Psi_j$.
Let $M$ be the boolean matrix with $A(M)=\alpha$.
It suffices to prove that 
\[
   d_{\Psi_j}(M) \ \ = \ \ 36 + 4\sqrt{2}
\]
is true for every $j\in[1,m]$.
To this end, fix an arbitrary $j\in[1,m]$ and let $\Psi\deff\Psi_j$.
Let $I=\set{i_1,i_2,i_3}\subseteq [1,k]$, and for each $i\in I$ let $L_i$ be a literal over $V_i$, such
that $\Psi=(L_{i_1}\oder L_{i_2}\oder L_{i_3})$.

Since $\alpha$ is a 1-in-3 model of $\Psi$, it satisfies a conjunctive
3-clause $\Gamma$ that is implicit in $\Psi$. 
W.l.o.g., 
\[
  \Gamma\ \ = \ \ (L_{i_1}\und\nicht L_{i_2}\und\nicht L_{i_3}).
\]
Let $\gamma:=\theta(\Gamma)$, and
for each $i\in I$ let $l_i\in\set{e_i,-e_i}$ be such that
$\gamma=\sum_{i\in I}l_i$.

Since $\alpha$ satisfies $\Gamma$, and since $M=A^{-1}(\alpha)$, it is straightforward to verify along the definition of the mappings $A$ and $\theta$ that $Me_i=l_i$ is true for each $i\in I$.
From Lemma~\ref{lem:Neu} we therefore obtain that $d_\Psi(M)=36+4\sqrt{2}$.
\end{proof}

Note that Lemma~\ref{lem:dphi} establishes the goal formulated
directly before equation~\eqref{eq:d1orth}:
When choosing 
\[
c^\orth(\Phi) \ \ \deff \ \ m\cdot (36+4\sqrt{2})
\]
whenever $\Phi$ is a 3-cnf formula consisting of $m$ disjunctive 3-clauses,
Lemma~\ref{lem:dphi} tells us that $\Phi$ has a 1-in-3 model if, and
only if, $d_1^\orth(S_\Phi,T_\Phi)=c^\orth(\Phi)$.

\paragraph{Relating $d_1^C(\bar{S}_\Phi,\bar{T}_\Phi)$ with 1-in-3 models of $\Phi$}
Until now, we only considered transformations using orthogonal matrices.
However, the congruence distance $d_1^C$ allows distance minimization
also by translating with an arbitrary vector.
The following lemma considers these transformations of time series too (cf. the appendix for a proof).

\begin{lemma}
    \label{lem:notranslation}
    Let $S, T$ be two time series of the same length over $\mathbb R^k$, 
    let $\bar S:=S\times -S$ and $\bar T:=T\times -T$.
    The following is true for every $M\in\orth(k)$ and every $v\in\RR^k$:
    \begin{align*}
        d_1(\bar S,\ M\cdot\bar T) \ \ \leq  \ \ d_1(\bar S,\ M\cdot\bar T + v)
    \end{align*}
\end{lemma}

As a consequence of Lemma~\ref{lem:dphi},
Lemma~\ref{lem:notranslation}, and the definition of
$d_1^C$, we immediately obtain the following (cf. the appendix for a proof).
\begin{theorem}\label{cor:dcongruence}
  Let $\Phi=\Oder_{j=1}^m\Psi_j$ be a 3-cnf formula with $m$
  disjunctive clauses. Then, $\Phi$ has a 1-in-3 model \ iff
  \[
     d_1^C (\bar S_\Phi, \bar T_\Phi) \ \ = \ \ m\cdot( 72+ 8\sqrt 2).
  \]
\end{theorem}

\paragraph{An algorithm solving \Problem{1-in-3-Sat}}

\begin{proof}[of Theorem~\ref{thm:reduction}] \ \\
Assume that $\mathcal{A}$ is an algorithm which, on input of two time
series $S$ and $T$ of equal length, computes $d_1^C(S,T)$.
Using this algorithm, the problem \Problem{3-in-1-Sat} can be solved
as follows.

Upon input of a 3-cnf formla $\Phi$,
construct the time series $\bar S_\Phi$ and $\bar T_\Phi$. Clearly,
this can be done in time polynomial in the size of $\Phi$. 
Letting $k$ be the number of variables occurring in $\Phi$, run
algorithm $\mathcal{A}$ with input $k,\bar S_\Phi, \bar T_\Phi$. 
After a number of steps polynomial in the size of $\Phi$, 
$\mathcal{A}$ will output (a suitable representation of) the number 
$d_1^C(S,T)$. 
Now, check if this number is equal to (a suitable representatio of)
the number  $m\cdot(72+ 8\sqrt{2})$, where $m$ is the number of
disjunctive clauses of $\Phi$.
If so, output ``yes''; otherwise output ``no''.

From Theorem~\ref{cor:dcongruence} we know that the algorithm's output
is ``yes'' if, and only if, $\Phi$ has a 1-in-3 model.
Thus, we have constructed a polynomial-time algorithm solving the
$\textup{NP}$-complete problem \Problem{1-in-3-Sat}. In case that $\textup{P}\neq \textup{NP}$, such an
algorithm cannot exist.
\end{proof}

Note that according to the above proof, already the restriction of
$d_1^C\textsc{-Computation}$ to input time series over
\[
  \set{0,1,-1,6,-6}^k
\]
cannot be accomplished in polynomial time,
unless $\textup{P}=\textup{NP}$.

\subsection{The structure $\Delta S$ of a time series $S$}
\label{sec:deltadistance}

In this subsection we consider the
well-known \emph{self-similarity matrix} of a time series.
Usually, the self-similarity matrix is used to analyze a time series
for patterns (e.\,g. using Recurrence Plots \cite{recurrenceplots}).

The important property that makes the self-similarity matrix useful
for approximating the congruence distance, is its invariance under 
transformations considered for the congruence distance, i.\,e. rotation, translation, and mirroring.

Considering an arbitrary time series $T=(t_0,\dots,t_{n-1})\in\mathcal
T_\MM$ over a metric space $(\MM, \dE)$, 
the \emph{self-similarity matrix}
\begin{displaymath}
    \Delta T \ \ \coloneqq \ \ \big(\; \dE\left(t_i,t_{i+j}\right)\; \big)_{i\in[0,n-1),\ j\in[1,n-i)}
\end{displaymath}
describes the inner structure of the time series.
Thus, we also call the self-similarity matrix $\Delta T$ the \emph{structure} of the time series $T$.

Throughout the remainder of this subsection, we will restrict attention to time
series over $(\RR^k,\dE_2)$. 

The next theorem shows that for such
time series, the structure $\Delta T$ completely describes the
sequence $T$ up to congruence, i.e., up to rotation, translation, and mirroring of
the whole sequence in $\mathbb R^k$. 

\begin{theorem}
    \label{thm:delta}
    Consider the metric space $(\mathbb R^k, \dE)$ for \linebreak[4] 
    $\dE:= \dE_2$,
    and let $S,T$ be two time series of length $n$ over $\mathbb R^k$.
    Then, $S$ and $T$ are congruent iff they have the same structure, i.\,e.:
    \begin{align*}
      S\cong_C T \ \ \iff \ \ \Delta S=\Delta T.
    \end{align*}
\end{theorem}
Basically, this theorem holds because the Eucledian Distance is invariant under isometric functions (cf. the appendix for a detailed proof).

\subsection{The Delta Distance}
\label{sec:bounddelta}

Our approach for approximating the congruence distance between
two time series $S$ and $T$ is to compare the self-similarity matrices
of $S$ and $T$ via a suitable matrix norm. This is formalized in the
following definition.

\begin{definition}[Delta Distance]
    \label{def:deltadistance}\upshape \ \\
    Let $\mathcal{T}$ be the class of all time series over $\RR^k$,
    and let $\Norm{\cdot}$ be a matrix norm.
    Let $S,T\in\mathcal{T}$ be two time series of length $m$ and
    $n$ ($m\leq n$), respectively.
    The \emph{delta distance} $d^\Delta_{\Norm{\cdot}}(S,T)$ is
    defined as follows: 
    \begin{align*}
        d^\Delta_{\Norm{\cdot}}(T,S) \coloneqq \ \
        d^\Delta_{\Norm{\cdot}}(S,T) \coloneqq \ \ 
        \min_{b\in[0, n-m]} \big(\, d^\Delta_{\Norm{\cdot}}(S, T_b^m) \,\big).
    \end{align*}
\end{definition}

We will consider the cases where $\Norm{\cdot}$ is the max column norm
$\Norm{\cdot}_m$ or the $p$-Norm $\Norm{\cdot}_p$ for some
$p\in\RR_{\geq 1}$. In these cases we will write $d^\Delta_m$ and
$d^\Delta_p$, respectively, to denote $d^\Delta_{\Norm{\cdot}}$.

Obviously, for time series of the same length, 
the complexity of computing the delta distance $d^\Delta_{\Norm{\cdot}}$
grows quadratically with the length of the time series. In particular,
for time series $S$ and $T$ of equal length,
$d^\Delta_1(S,T)$ and $d^\Delta_m(S,T)$ can be computed in time
quadratic in the length of $S$ and $T$.
\medskip

Our next aim is to show that the the delta distance $d^\Delta_m$
provides a lower bound on the congruence distance $d_1^C$, as
formulated in the following theorem.

\begin{theorem}
    \label{thm:quadmaxest}
   For all time series $S$ and $T$ over $\RR^k$, the following holds:
    \begin{align*}
        d^\Delta_m(S,T) \ \ \leq \ \ 2\cdot d^C_1(S,T).
    \end{align*}
\end{theorem}

For proving Theorem~\ref{thm:quadmaxest}, we will emply the following
two lemmas (cf. the appendix for their proofs).

\begin{lemma}
    \label{lem:minwindow}
    Let $\mathcal{T}$ be the set of all time series over $\RR^k$.
    Let $\Norm{\cdot}$ be a matrix norm.
    Let $p\in\RRposs{1}$, and
    let $C$ be a function from $\NN$ to
    $\RR$.
    If for all $n\in\NN$ and all time series $S,T\in\mathcal{T}$ of
    length $n$ we have
    \begin{align*}
        d_{\Norm{\cdot}}^\Delta(S,T) \ \leq \ C(n)\cdot d_p(S,T),
    \end{align*}
    then
    \begin{align*}
        d_{\Norm{\cdot}}^\Delta(S,T) \ \leq \ C(\min\left\{ \#S,\#T \right\})\cdot d_p(S,T)
    \end{align*}
    holds for all time series $S,T \in\mathcal{T}$ (i.e., also for
    time series of \emph{different} lengths).
\end{lemma}

\begin{lemma}
    \label{lem:quadmaxest}
    The following holds 
    for the max column norm
    $\|\cdot\|_m$ and for all time series $S,T$ over $\RR^k$:
    \begin{align*}
        d^\Delta_m(S,T) \ \ \leq \ \ 2\cdot d_1(S,T).
    \end{align*}
\end{lemma}

\begin{proof}[of Theorem~\ref{thm:quadmaxest}] \ \\
    W.\,l.\,o.\,g., let $\#S = m\leq n = \#T$.
    For arbitrary $M\in\orth(k)$ and $v_0\in\mathbb R^k$,
    Theorem~\ref{thm:delta} tells us that
    \begin{align*}
        \Delta T \ \ = \ \  \Delta( M\cdot T+v_0 ).
    \end{align*}
    Hence, rotation, translation, and mirroring of $T$ does not affect
    $d^\Delta_m(S,T)$.
    Applying Lemma~\ref{lem:quadmaxest}, we obtain
    \begin{align*}
        d^\Delta_m(S,T) = d^\Delta_m(S,M\cdot T+v_0) \leq 2\cdot d_1(S,M\cdot T+v_0)
    \end{align*}
    Thus, the desired inequality holds:
    \begin{align*}
        d^\Delta_m(S,T) & \leq \inf_{M\in\orth(k), v_0\in\mathbb R^k}
        \big( \, 2\cdot d_1(S,M\cdot T+v_0) \,\big) \\
        &= 2\cdot d^C_1(S,T)
    \end{align*}
\end{proof}

Similarly to Lemma~\ref{lem:quadmaxest}, we can also prove the
following (cf.\ the appendix for a proof).

\begin{lemma}
    \label{lem:quadsumest}
    \hfill The following holds for the matrix norm $\|\cdot\|_1$ and
    for all time series $S,T$ over $\RR^k$:
    \begin{align*}
        d^\Delta_1(S,T) \ \ \leq \ \ \big(\min\left\{ \#S,\#T
        \right\}-1) \cdot d_1(S,T).
    \end{align*}
\end{lemma}

Combining this lemma with the proof of 
    Theorem~\ref{thm:quadmaxest}, we obtain the following.

\begin{theorem}
    \label{thm:quadsumest}
   For all time series $S$ and $T$ over $\RR^k$, the following holds:
    Then the following inequality holds:
    \begin{align*}
        d^\Delta_1(S,T) \ \ \leq \ \ \big(\min\left\{ \#S,\#T \right\} -1\big)\cdot d^C_1(S,T).
    \end{align*}
\end{theorem}

The Theorems~\ref{thm:quadsumest} and \ref{thm:quadmaxest} show that the delta distances $d^\Delta_1$ and $d^\Delta_m$ provide lower bounds for the congruence distance $d^C_1$.
On the other hand, the ratio of the congruence distance and the delta distance can grow arbitrarily as shown with the following example.

\begin{example}
    Consider $\MM=\RR^2$ and the eucledian distance $\dE=\dE_2$.
    We show that for each $C>0$ time series $S,T$ exist such that $\frac{\delta^C(S,T)}{\delta^\Delta(S,T)}\geq C$ (cf. Figure~\ref{fig:noupperbound}).
    \begin{figure}
        \centering
        \includegraphics{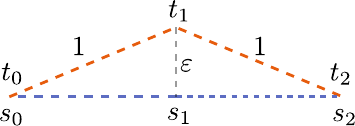}
        \caption{Two time series $S$ and $T_\epsilon$ in $\mathbb R^2$, such that $\frac{d_1^C(S,T_\epsilon)}{d_1^\Delta(S,T_\epsilon)}\xrightarrow{\varepsilon\rightarrow 0}+\infty$.}
        \label{fig:noupperbound}
    \end{figure}
    Let $\epsilon>0$, $a\deff\sqrt{1-\epsilon^2}$, and consider $S=(s_0,s_1,s_2),T_\epsilon=(t_0,t_1,t_2)\in\TS_3$ with
    \begin{align*}
        S\deff \left( \colvec{2}{-a}{0}, \colvec{2}{0}{0}, \colvec{2}{a}{0} \right),\ 
        T_\epsilon\deff \left( \colvec{2}{-a}{0}, \colvec{2}{0}{\epsilon}, \colvec{2}{a}{0} \right).
    \end{align*}
    Then $\delta^\Delta(S,T_\epsilon)=2(1-a)=2(1-\sqrt{1-\epsilon^2})$.
    We claim that $\delta^C(S,T_\epsilon)\geq\frac \epsilon 2$ (cf. the Appendix~A, Claim~\ref{cl:noupperbound} for a proof).
    Then,
    \begin{align*}
        \frac{\delta^C(S,T_\epsilon)}{\delta^\Delta(S,T_\epsilon)} \geq \frac 1 4 \cdot \frac{\epsilon}{1-\sqrt{1-\epsilon^2}} \xrightarrow{\quad\epsilon\rightarrow 0\quad}+\infty.
    \end{align*}
\end{example}


%% file: reducingcomplexity.tex
\section{The Reduced Delta Distance}
\label{sec:approximatedcongruence}

While computing the congruence distance $d^C_1$ of two given time
series $S$ and $T$ is an $\textup{NP}$-hard problem, the computation
of the delta distances $d^\Delta_m$ and $d^\Delta_1$ can be
accomplished in time quadratic in the lengths of $S$ and $T$.
For practical usage, however, a distance measure that can be computed
in linear or quasi-linear time, would be highly desirable.

In this section, we propose the $d^\delta$ distance function, for
which the distance between two time series of length $n$ can be
computed in time $\mathcal O(n\log n)$.
The idea underlying the definition of $d^\delta$ is the same
as for the delta distance function $d^\Delta_{\Norm{\cdot}}$, except
that only $\log n$ columns of the matrix 
$\Delta S-\Delta T$ are computed.

For giving the precise definition of the reduced delta distance
$d^\delta$, we need the following notation.

We write $\PowerOfTwo$ for the set $\setc{2^n}{n\in\NN}$ of all powers
of $2$.
Let $T=\left(t_0,\cdots,t_{n-1}\right)\in\mathcal T$ be a time series over the metric space $(\RR^k,\dE)$.
The \emph{reduced structure} of $T$ is the matrix
\begin{align*}
    \delta T \ \ \coloneqq \ \ \big( \dE\left(t_i,t_{i+j}\right) \big)_{i\in[0,n-1),\ j\in[1,n-i)\,\cap\, \PowerOfTwo}
\end{align*}
It contains the subset of the entries of $\Delta T$ which compare two
states $t_i$ and $t_{i+j}$ having a time distance $j$ that is a power of 2.

\begin{definition}[Reduced Delta Distance]
    \label{def:reduceddeltadistance} \upshape \ \\
    Let $\mathcal{T}$ be the class of all time series over $\RR^k$,
    and let $\Norm{\cdot}$ be a matrix norm. Let $S,T\in\mathcal{T}$
    be two time series of lengths $m$ and $n$ ($m\leq n$), respectively. The     
    \emph{reduced delta distance} $d^\delta_{\Norm{\cdot}}(S,T)$ is
    defined as follows: 
    \begin{align*}
        d^\delta_{\Norm{\cdot}}(T,S) \coloneqq \ \
        d^\delta_{\Norm{\cdot}}(S,T) \coloneqq \ \ 
        \min_{b\in[0, n-m]} \big(\, d^\delta_{\Norm{\cdot}}(S, T_b^m) \,\big).
    \end{align*}
\end{definition}
In case that $\Norm{\cdot}$ is the max column norm $\Norm{\cdot}_m$ or
the $p$-Norm $\Norm{\cdot}_p$ for some $p\in\RRposs{1}$, we 
will write $d^\delta_p$ and $d^\delta_m$, respectively, to denote
$d^\delta_{\Norm{\cdot}}$.

In particular, since $\delta T$ has $\mathcal{O}(n\log n)$ entries, the values
$d^\delta_m(S,T)$ and $d^\delta_1(S,T)$ can be computed in
time $\mathcal{O}(n\log n)$, if $S$ and $T$ are two time series of length $n$.
Thus, using $d^\delta$ (instead of $d^\Delta$ or $d^C_1$) has the benefit that
the distance between two time series of equal length can be computed
in quasi-linear time.

On the other hand, using $\delta T$ instead of $\Delta T$ has the drawback that Theorem~\ref{thm:delta} (i.e.,
the congruence requirement) does not hold for $\delta T$:
The following example shows that
there are time series $S,T\in\mathcal T$ with $\Delta S\ne \Delta T$ but $\delta S = \delta T$.

\begin{example}
    \label{ex:nocong}
    Consider the following time series $S,T\in\mathcal T$ over the metric space $(\mathbb R^2,\dE_2)$:
    \begin{align*}
        S &= \textstyle \left( \colvec{2}{-4}{0}, \colvec{2}{0}{0}, \colvec{2}{0}{3}, \colvec{2}{-4}{0} \right) \\
        T &= \textstyle \left( \colvec{2}{-4}{0}, \colvec{2}{0}{0}, \colvec{2}{0}{3}, \colvec{2}{4}{0} \right)
    \end{align*}
    Then
    \begin{align*}
        \begin{pmatrix}
            4 & 5 & 0 \\
            3 & 4 \\
            5 &
        \end{pmatrix}
        \ \ = \ \
         \Delta S \ \ \ne \ \ \Delta T \ = \ \begin{pmatrix}
            4 & 5 & 8 \\
            3 & 4 \\
            5 &
        \end{pmatrix}
    \end{align*}
    but
    \begin{align*}
        \delta S \ \ = \ \ \begin{pmatrix}
            4 & 5 \\
            3 & 4 \\
            5 &
        \end{pmatrix}
        \ \ = \ \ \delta T.
    \end{align*}
\end{example}

\subsubsection*{Cheap Lower Bound for Congruence}

In Section~\ref{sec:bounddelta} we showed that the congruence distance yields an upper bound for the delta distance.
Analogously, in this section we show that the congruence distance
yields an upper bound for the reduced delta distance.
Viewed from the other side,
the reduced delta distance function can thus be regarded as a
computationally cheap approximation of the congruence distance, which
provides a lower bound for the congruence distance.

\begin{theorem}
    \label{thm:qlinmaxest}
    For all time series $S$ and $T$ over $\RR^k$, the following holds:
    \begin{align*}
        d^\delta_m(S,T) \ \ \leq \ \ 2\cdot d^C_1(S,T).
    \end{align*}
\end{theorem}
\begin{proof}
    Obviously,
    $d^\delta_m(S,T)\leq d^\Delta_m(S,T)$ holds for all time series
    $S,T$ over $\RR^k$.
    Hence, the desired inequality follows with Theorem~\ref{thm:quadmaxest}.
\end{proof}

The proof of the following theorem is much more algebraic and can be found in the appendix.
\begin{theorem}
    \label{thm:qlinsumest}
    For all time series $S$ and $T$ over $\RR^k$, the following holds:
    \begin{align*}
        d^\delta_1(S,T) \ \ \leq \ \ \lfloor2\cdot\log(\min\left\{ \#S,\#T \right\}-1)\rfloor\cdot d^C_1(S,T).
    \end{align*}
\end{theorem}

%% file: experiments.tex
\section{Experimental Results}
\label{sec:experiments}

Although we have boundaries for the delta distance function, the ratio of $d_1^C$ to $d_p^\Delta$ could be arbitrarily bad.
With the experiments below, we show that the ratio has a stable average.
We consider the delta distance functions using the max column matrix norm only (i.\,e. we focus on $d_m^\Delta$ and $d_m^\delta$) since we achieved best results with it.

\begin{algorithm}
    \begin{lstlisting}[mathescape]
Algorithm: gen
Input:
    $S = (s_0, \cdots, s_{n-1}),$
    $\eta>1,$
    $E\in\mathbb N,$
Output: $\texttt{gen}(S)$
repeat $E$ times:
    Choose a random subset $I\subseteq\left\{ 0,\dots,n-1 \right\}$
    Calculate the barycenter $b\coloneqq |I|^{-1}\sum_{i\in I} s_i$
    Choose a random value $\mu\in\left[ \eta^{-1}, \eta \right]$
    for each $i\in I$:
        $s_i\coloneqq b + \mu\cdot( s_i-b )$
    end
end
return $S$
    \end{lstlisting}
    \caption{The algorithm $\texttt{gen}$ generates new time series from a given time series such that it is likely to be aligned to its origin. The parameters $\eta$ and $E$ are considered to be clear from the context.}
    \label{alg:gen}
\end{algorithm}

We performed the experiments using the TRECVID benchmark dataset \cite{TRECVID} which consists of a set of videos, denoted by $\trecvid$, and consider them as high demensional time series.
The dataset consists of around 8000 videos from $13$\,sec to 2\,h.
We downscaled them to gray-scale videos with a resolution of $128\times 96$ Pixels and a framerate of two images per second, i.\,e. we consider the metric space $(\MM=\RR^{12,288},\dE)$ with $|M|=8000$.
Each video $T=(t_0,\cdots,t_{n-1})\in\trecvid$ is considered as a time series with $26\leq\#T\leq 14,400$, and with each image $t_i\in\RR^{12,288}$ considered as a $12,288$ dimensional vector with each pixel corresponding to one dimension.

Since it is too complex to compute the exact value of $d_1^C$, we used a transformation function $g:\trecvid\longrightarrow \TS$ to generate new time series for each time series $S\in\trecvid$, and computed the distance $d_1(S,\texttt{gen}(S))\geq d_1^C(S,\texttt{gen}(S))$.
Intuitively, $g$ creates random time series such that for a certain time series, the time series distance to the random time series is close to their congruence distance.
Technically, $g$ explodes or implodes random states of the time series around their barycenter.
The exact algorithm is shown in Algorithm~\ref{alg:gen}.
Hence, we can provide an upper bound for the relative error of $d_m^\Delta$ and $d_m^\delta$:
\begin{align*}
    e^\Delta_m(S,\texttt{gen}(S)) \coloneqq&\ \frac{d_1(S,\texttt{gen}(S))}{2\cdot d_m^\Delta(S,\texttt{gen}(S))} &\ \geq \frac{d_1^C(S,\texttt{gen}(S))}{2\cdot d_m^\Delta(S,\texttt{gen}(S))} \\
    e^\delta_m(S,\texttt{gen}(S)) \coloneqq&\ \frac{d_1(S,\texttt{gen}(S))}{2\cdot d_m^\delta(S,\texttt{gen}(S))} &\ \geq \frac{d_1^C(S,\texttt{gen}(S))}{2\cdot d_m^\delta(S,\texttt{gen}(S))}
\end{align*}

Table~\ref{tab:experiments} shows the average ratios $E[e^\Delta_m(S,\texttt{gen}(S))]$ and $E[e^\Delta_m(S,\texttt{gen}(S))]$ for all time series and their generated time series for $12$ sets of parameters of the algorithm $\texttt{gen}$.
A closer inspection of the algorithm $\texttt{gen}$ shows that the congruence distance increases with increasing parameter $\eta$.
Thus, the results from Table~\ref{tab:experiments} suggest that the boundary for the relative error decreases with increasing congruence distance.
Note, that we only approximate the relative error from top.
We did not investigate the variation of the error for different numbers of explosions $E$, but it could probably vary as a result of the deviation of $d_1(S,\texttt{gen}(S))$ from $d_1^C(S,\texttt{gen}(S))$.

\begin{table}
    \centering
    \begin{tabular}{|l|r|r|r|r|}
        \hline
        E \textbackslash \ $\eta$  & $1.1$     & $2$       & $10$      \\
        \hline
        $1$         & $1.62$    & $1.53$    & $1.32$      \\
        $5$         & $1.82$    & $1.62$    & $1.22$       \\
        $10$        & $1.78$    & $1.62$    & $1.22$        \\
        $50$        & $1.85$    & $1.35$    & $1.13$         \\
        \hline
    \end{tabular}
    \hfill
    \begin{tabular}{|l|r|r|r|r|}
        \hline
        E \textbackslash \ $\eta$  & $1.1$     & $2$       & $10$      \\
        \hline
        $1$         & $1.68$    & $1.58$    & $1.35$      \\
        $5$         & $1.89$    & $1.65$    & $1.23$       \\
        $10$        & $1.85$    & $1.65$    & $1.24$        \\
        $50$        & $1.90$    & $1.37$    & $1.13$         \\
        \hline
    \end{tabular}
    \caption{Experimental results: The average $E[e^\Delta_m(S,\texttt{gen}(S))]$ (left) and $E[e^\delta_m(S,\texttt{gen}(S))]$ (right) for different parameter settings of $\texttt{gen}$.}
    \label{tab:experiments}
\end{table}

Figure~\ref{fig:experiments} shows the results in more detail for fixed parameters $\eta=1.5$ and $E=5$.
The experimantal results give evidence for the presumption that the congruence distance function is more accurate on time series having a large congruence distance.
On the other hand, the relative error increases for time series having a low congruence distance.

\begin{figure}
    \centering
    \includegraphics[width=.45\textwidth]{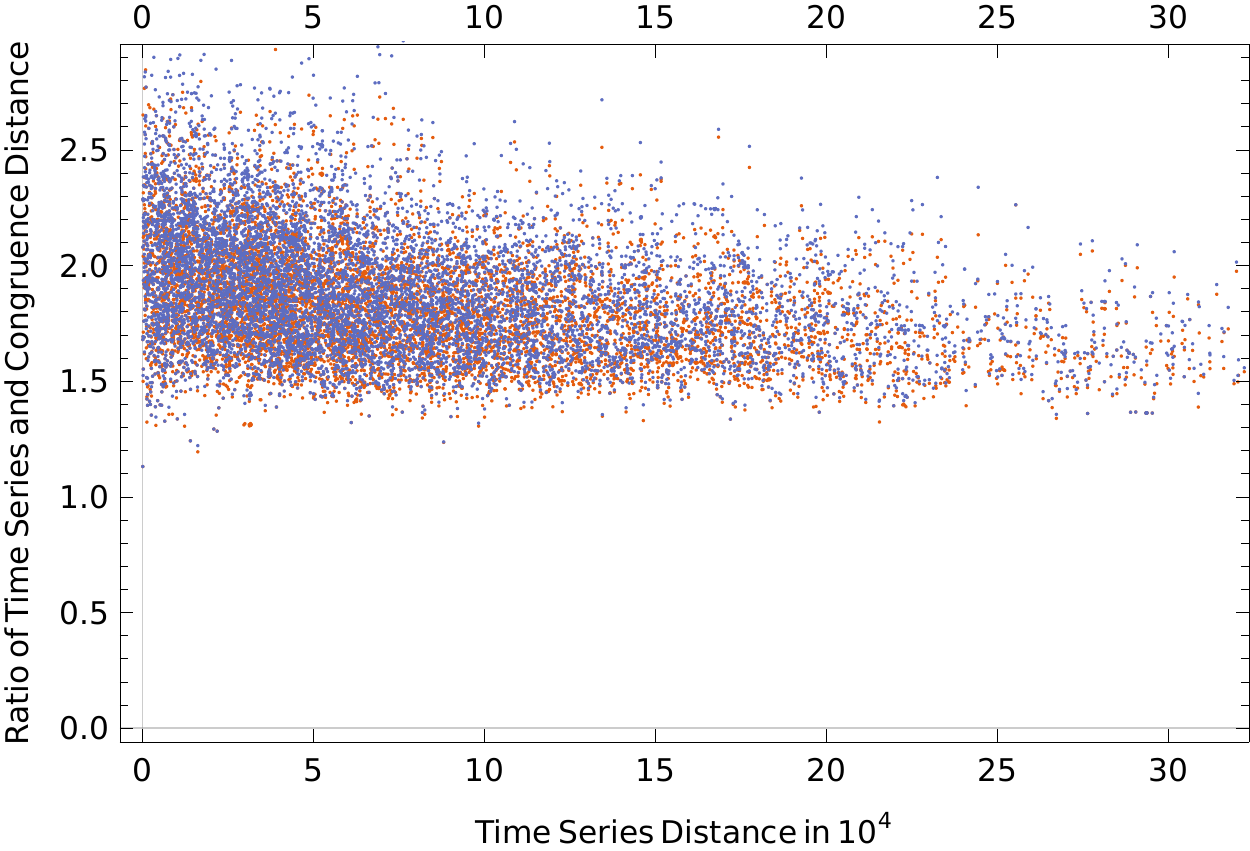}
    \caption{Experimental Results showing ratio $e^\Delta_m(S,\texttt{gen}(S))$ (orange) and $e^\delta_m(S,\texttt{gen}(S))$ (blue) for each $S\in\trecvid$; $\eta=1.1$, $E=5$; $x$-axis: $d_1^C$; $y$-axis: ratio of time series distance and (reduced) delta distance.}
    \label{fig:experiments}
\end{figure}


%% file: conclusion.tex
\section{Conclusion and Future Work}
\label{sec:conclusion}

In this paper, we introduced and analyzed the problem of measuring the congruence between two time series.
After having proved that its computation is $\textup{NP}$-hard, 
we provided two measures for approximating the congruence distance in polynomial time.
The first (namely, the delta distance) measures the congruence in a way that the distance of two time series is $0$ iff they are  congruent.
The second loses this benefit, but can be computed in quasi-linear instead of quadratical time.
Furthermore, we showed that all provided distances fulfill the triangle inequality on time series of the same length.

We simplified the problem of motion gesture recognition to measure the congruence of two time series in $3$ dimensional space.
In practical applications, multiple time series at the same time need to be considered, which is one of our next steps to investigate.
Furthermore, the congruence distance is not robust against scaling of the time series, which is an important future work too.
Also, we are currently carrying out experiments on motion gesture recognition and content based video copy detection to evaluate the utility of the congruence distance.
To achieve more robust distance functions, local time shifting techniques like Dynamic Time Warping need to be adapted to the congruence distance provided in this paper.

\bibliographystyle{plain}
\bibliography{gesture,cbcd,congruence,datastructures,distances}

%% file: appendix.tex

\onecolumn

\newpage

\appendix

\section{Details Omitted in Section~4}

\begin{lemma}
    Definition~\ref{def:congruencedistance} is well-defined.
\end{lemma}
\begin{proof}
    Consider arbitrary but fixed $S,T\in\mathcal T_n$, $M\in\orth(k)$,
    and $v\in\RR^k$.
    We write $\bar v$ and $\bar 0$ for the time series
    $(y_0,\ldots,y_{n-1})$ and $(z_0,\ldots,z_{n-1})$ in
    $\mathcal{T}_n$, where $y_i=v$ and $z_i=0\in\RR^k$,
    for each $i\in[0,n)$. 
    
    It is not difficult to verify that the following is true:
    \begin{align*}
        d_p(S,\ M\cdot T+v) \ \geq \ \ & d_p(M\cdot T+v,\ \bar 0) -
        d_p(S,\bar 0) \\
        \geq \ \ & d_p(\bar v,\ \bar 0)-d_p(M\cdot T,\ \bar 0)-d_p(S,
        \bar 0) \\
        = \ \ & d_p(\bar v, \bar 0) - d_p(T,\bar 0) - d_p(S,\bar 0)
    \end{align*}
    Note that $d_p(\bar v,\bar 0) = \Norm{\dE_2(\bar v,\bar 0)}_p =
    \Norm{\big( \Norm{v}_2 \big)_{i\in[0,n)}}_p = (n \cdot
    \Norm{v}_2^p\big)^{1/p} = n^{1/p}\cdot \Norm{v}_2$.
    Thus, we have
    \[
       d_p(S,\ M\cdot T+v) \ \geq \ \ n^{1/p}\cdot \Norm{v}_2 \ - \
       d_p(T,\bar 0) - d_p(S,\bar 0).
    \]
    Let 
    \[
       I \ := \ \  \inf_{M'\in\orth(k),\, v'\in\mathbb R^k} d_p\left(S,\,M'\cdot T+v'\right).
    \]
    Clearly, $I\leq d_p(S,T)$ (for this, consider $v'=0\in\RR^k$ and
    $M'$ the identity matrix).
    Now, let $\epsilon >0$, and let
    \[
       r \ \coloneqq \ \ \frac{d_p(S,T) + d_p(T,\bar 0) + d_p(S,\bar
       0) + \varepsilon}{n^{1/p}}.
    \]
    Then, for any $v\in\RR^k$ with $\Norm{v}_2>r$, we have
    \[
      d_p(S,\ M\cdot T+v) \ \geq \ I + \epsilon, \quad \text{for
        every $M\in\orth(k)$}.
    \]
    Thus, for computing $I$ it suffices to restrict attention to
    $v'\in\RR^k$ with $\Norm{v'}_2\leq r$. I.e.,
    \[
       I \ = \ \ \inf_{M'\in\orth(k),\, v'\in V} d_p\left(S,\,M'\cdot T+v'\right),
    \]
    for $V\coloneqq \left\{ v\in\RR^k :  \|v\|_2\leq r \right\}$.
 
    We let $X\deff \orth(k)\times V$ and $Y\deff \RRpos$ and let
    $f:X\longrightarrow Y$ be defined via
    \[
       f(M,v) \ \deff \ d_p(S,\ M\cdot T + v).
    \]
    Furthermore, view each element $(M,v)\in X$ as a vector in
    $\RR^{k^2+k}$, and choose the Eucledian distance
    $\dE_2$ as a metric on $X$. It is straightforward to see that,
    with respect to this metric,
    \begin{enumerate}[(1)]
      \item $X$ is a bounded set, and
      \item $f$ is a continuous function.
    \end{enumerate}

  \noindent
    Now consider an arbitrary sequence $\xi\coloneqq (M_i,v_i)_{i\in\NN}$  with $M_i\in\orth(k)$ and $v_i\in V$ for all $i\in\NN$ such that
    \begin{align*}
      I \ = \ \lim_{i\rightarrow\infty} d_p(S,\ M_i\cdot T+v_i) \ \ =
      \ \ \lim_{i\rightarrow\infty}f(M_i,v_i). 
    \end{align*}
    Since $X=\orth(k)\times V$ is a bounded set (w.r.t.\ $\dE_2$), the sequence $\xi$ is
    bounded, and thus contains a convergent (w.r.t.\ $\dE_2$) subsequence
    $\xi'=(M_{i_j},v_{i_j})_{j\in\NN}$, with $i_1<i_2<\cdots$ (Bolzano-Weierstrass Theorem).
    Let $(M^*,v^*)\in \orth(k)\times \RR^k$ be the limit, i.\,e.,
    \begin{align*}
        (M^*,v^*) \ \coloneqq \ \ \lim_{j\rightarrow\infty} (M_{i_j},v_{i_j}).
    \end{align*}

   \noindent
   Since $f$ is a continuous function, we obtain that
   \ $\displaystyle
     \lim_{j\rightarrow\infty} f(M_{i_j},v_{i_j}) \ \ = \ \ f(M^*,v^*).
   $ \
   Thus, 
   \[
      I \ = \ \lim_{i\rightarrow\infty} f(M_i,v_i) 
        \ = \ \lim_{j\rightarrow\infty} f(M_{i_j},v_{i_j})
        \ = \ f(M^*,v^*)
        \ = \ d_p(S,\ M^*\cdot T+v^*).
   \]
    Hence, Definition~\ref{def:congruencedistance} is well-defined.
\end{proof}

\bigskip

\begin{proof}[of Proposition~\ref{prop:congruencetriangle}] \ \\
    It is easy to see that $d_p$ and thus $d_p^C$ are symmetric functions, i.e.,  $d_p^C(S,T)=d_p^C(T,S)$.
    Furthermore, the triangle inequality for $d_p$ follows from the axioms of a norm.

    To prove the triangle inequality of $d_p^C$, let
    $M_1,M_2\in\orth(k)$ and $v_1,v_2\in\RR^k$ such that $d_p^C(S,
    T)=d_p(M_1\cdot S+v_1, T)$ and $d_p^C(T,U)=d_p(T, M_2\cdot U+v_2)$.
    Then, the triangle inequality follows:
    \begin{align*}
        d_p^C(S,U)\ \leq \ &\, d_p(M_1\cdot S+v_1, M_2\cdot U+v_2) \\
        \ \leq \ &\, d_p(M_1\cdot S+v_1, T) + d_p(T, M_2\cdot U+v_2) \\
        \ = \ &\, d_p^C(S,T)+d_p^C(T,U).
    \end{align*}
\end{proof}

\bigskip

\begin{proof}[of Lemma~\ref{lem:d'}] \ \\
Recalling that $d_1=\Norm{\dE_2}_1$, for the particular choice of
$S'_\Psi$ and $T'_\Psi$ we obtain that
\begin{equation}\label{eq:lem:d'}
  d_1(S'_\Psi,M\cdot T'_\Psi) \ = \ 
  6 \cdot \sum_{i\in I}\big( \;
     \dE_2(e_i,Me_i) + \dE_2(-e_i,Me_i)
  \,\big).
\end{equation}
Hence, $d_1(S'_\Psi,M\cdot T'_\Psi) = 6 \cdot \sum_{i\in I}(a_i+b_i)$.
\end{proof}

\bigskip

\begin{proof}[of Lemma~\ref{lem:Neu}] \ \\
Let $M$ be the $(k\times k)$-matrix with columns\\
$m_1,\ldots,m_k$ such
that $m_i=l_i$ for each $i\in I$, and $m_j=e_j$ for each
$j\in[1,k]\setminus I$.
It is straightforward to verify that $M\in\orth(k)$ and $Me_i=l_i$ for each $i\in I$.

Now let $M$ be an arbitrary element in $\orth(k)$ such that $Me_i=l_i$
for all $i\in I$.
Using the Lemmas~\ref{lem:d'} and \ref{lem:d}, we obtain that \
$
   d_\Psi(M)
 = d_1(S'_\Psi,M\cdot T'_\psi) + d_1(S_\Psi,M\cdot T_\Psi)
 = $
\[
  6\cdot \sum_{i\in I}(a_i+b_i) \ \ + \ \ 4\sqrt{2}. 
\]
Recall that $a_i=\dE_2(e_i,Me_i)$ and $b_i=\dE_2(-e_i,Me_i)$.
Since $Me_i=l_i\in\set{e_i,-e_i}$, we obtain that
\[
   a_i+b_i
   \ \ = \ \ 
   \dE_2(e_i,-e_i)
   \ \ = \ \  
   2\cdot \Norm{e_i}_2 
   \ \ = \ \ 2.
\]
Hence, \ $\displaystyle 6\cdot\sum_{i\in I}(a_i+b_i) = 6\cdot 3\cdot 2
= 36$, \ and \
$d_\Psi(M) = 36 + 4\sqrt{2}$.
\\ \mbox{}  
\end{proof}

\bigskip

\begin{proof}[of Lemma~\ref{lem:notranslation}] \ \\
Let $S,T,v$ be as in the lemma's assumption. Let us fix an arbitrary
$M\in\orth(k)$.
Let $T'\deff M\cdot T$. Let $S=(s_0,\ldots,s_{n-1})$ and let $T'=
(t_0,\ldots,t_{n-1})$.
Then,
\begin{equation*}
 \begin{array}{rcl}
 d_1(\bar S, M\cdot \bar T)
 & = &\displaystyle
 \Big( \sum_{i=0}^{n-1} \dE_2(s_i,t_i) \Big) \ + \
 \Big( \sum_{i=0}^{n-1} \dE_2(-s_i,-t_i)\Big)
 \smallskip\\
 & = &\displaystyle
 \sum_{i=0}^{n-1}\  
    2\cdot \Norm{s_i-t_i}_2,
\end{array}
\end{equation*}
and \ $ d_1(\bar S, M\cdot \bar T + v) \ =$
\begin{equation*}
 \sum_{i=0}^{n-1} \Big(\ 
  \Norm{s_i-t_i-v}_2 \ + \ \Norm{s_i-t_i+v}_2
  \ \Big)  
\end{equation*}
Letting \ $u_i\deff s_i-t_i$, \ we obtain that \
\[
 \begin{array}{rcl}
   d_1(\bar S, M\cdot \bar T)
 & = 
 & \displaystyle\sum_{i=0}^{n-1}\ 2\Norm{u_i}_2, \qquad\text{and}
\smallskip\\
   d_1(\bar S, M\cdot \bar T + v) 
 & =
 & \displaystyle\sum_{i=0}^{n-1} \Big(\, \Norm{u_i+v}_2 \ + \ \Norm{u_i-v}_2 \, \Big).
 \end{array}
\]  
For proving the lemma, it therefore suffices to show that 
\begin{equation}\label{eq:LemmaTranslation_v}
  2\Norm{u_i}_2 \ \ \leq \ \ 
  \Norm{u_i+v}_2 \ + \ \Norm{u_i-v}_2
\end{equation}
is true for every $i\in[0,n)$.
In what follows, we prove that the
inequality~\eqref{eq:LemmaTranslation_v} is in fact true for \emph{every}
vector $u_i\in\RR^k$.

For achieving this, the following claim will be useful.

\begin{claim}\label{claim:1}
 $|a+b|+|a-b|\geq 2|a|$ \ is true for all $a,b\in\RR$.
\end{claim}
\begin{proof}
   Let us first restrict attention to the case where $a\geq 0$.
   In this case, the following is true:
    \begin{align*}
        |a+b|+|a-b| =& \begin{cases}
            a+b+a-b \ \ \ \ \, = \ 2a,& \text{ if } a\geq |b| \\
            a+|b|-a+|b| \ = \ 2|b|,& \text{ if } a < |b|
        \end{cases}
    \end{align*}
    In both cases, $|a+b|+|a-b|\geq 2|a|$, and we are done.

   Now let us consider the case where $a<0$. Then, $a'\deff -a >0$,
   and
      \ $  |a+b|+|a-b| \  = \ |a'-b|+|a'+b|$.
   For the latter, we already know that it is $\geq 2|a'|=2|a|$.
  This completes the proof of Claim~\ref{claim:1}.
 \end{proof}
 
 Now, let $u$ be an arbitrary vector in $\RR^k$.
 Clearly, there exists an element $M'\in\orth(k)$ such that $M'u=
 ae_1$ for some $a\in\RR$. For this $M'$ let us write $b_i$ for the 
 entry in the $i$-th component of the vector $M'v$, for every $i\in[1,k]$.
 Then, the following is true:
 \[
    2\Norm{u}_2 \ = \ 2\Norm{M'u}_2 \ = \ 2|a|
 \]
 and
 \[
  \begin{array}{rcl}
    \Norm{u+v}_2 + \Norm{u-v}_2
  & = & \Norm{M'(u+v)}_2 + \Norm{M'(u-v)}_2
  \smallskip\\
  & = & \Norm{M'u+M'v}_2 + \Norm{M'u - M'v}_2
  \smallskip\\
  & = & \Norm{ae_1 + M'v}_2 + \Norm{ae_1 - M'v}_2
  \smallskip\\
  & \geq & |a+b_1| \ + \ |a-b_1|.
  \end{array}
 \]
 By Claim~\ref{claim:1}, the latter is \ $\geq 2|a|$.
 In summary, we obtain that 
 \[
    2\Norm{u}_2 \ \ \leq \ \ 
    \Norm{u+v}_2 \ + \ \Norm{u-v}_2
 \] 
 is true for all $u,v\in\RR^k$.
 \\
 This completes the proof of Lemma~\ref{lem:notranslation}.
\end{proof}

\bigskip

\begin{proof}[of Theorem~\ref{cor:dcongruence}] \ \\
 From Lemma~\ref{lem:notranslation} we obtain that
 \[
    d_1^C(\bar S_\Phi, \bar T_\Phi) 
    \ \ = \ \ 
    \min_{M\in\orth(k)} d_1(\bar S_\Phi,\ M\cdot \bar T_\Phi).
 \]
 Furthermore, by the choice of $\bar S_\Phi$ and $\bar T_\Phi$ we have
 for every $M\in\orth(k)$ that
 \[
   d_1(\bar S_\Phi,\ M\cdot \bar T_\Phi) 
   \ \ = \ \ 
   2\cdot d_1(S_\Phi,\ M\cdot T_\Phi).
 \]
 Hence, \ $d_1^C(\bar S_\Phi, \ \bar T_\Phi) \ = $
 \[
    2\cdot \min_{M\in\orth(k)} d_1(S_\Phi,\ M\cdot T_\Phi)
  \ \ \stackrel{\text{Lemma~\ref{lem:dphi}}}{=} \ \
    2\cdot m\cdot (36+4\sqrt{2}).
 \]
\end{proof}

\bigskip

\begin{proof}[of Theorem~\ref{thm:delta}] \ \\
Let \ $S=(s_0,\ldots,s_{n-1})$ \ and \ $T=(t_0,\ldots,t_{n-1})$.
For the direction ``$\Longrightarrow$'' 
assume that $T=M\cdot S+v_0$ for some orthogonal matrix $M\in\orth(k)$ and some vector $v_0\in\mathbb R^k$.
Then $\Delta S=\Delta T$, since the following equation holds for all $i,j\in[0,n)$:
    \begin{align*}
        \dE(t_i,t_j)^2 &= \dE(M s_i+v_0, \ M s_j+v_0)^2 \\
                     &= \| M s_i+v_0 - M s_j-v_0\|_2^2 \\
                     &= \| M( s_i-s_j )\|_2^2 \\
                     &= \| s_i-s_j \|_2^2 \\
                     &= \dE(s_i, s_j)^2
    \end{align*}

    For the opposite direction, assume $\Delta S=\Delta T$.
    Let $S'=(0,s_1',\ldots,s_{n-1}')=S-s_0$ and
    $T'=(0,t_1',\ldots,t_{n-1}')=T-t_0$.
    Obviously, $\Delta S'=\Delta S$ and $\Delta T'=\Delta T$. Hence,
    $\Delta S'=\Delta T'$, and we thus obtain the equality of scalar products:
    \begin{align*}
        \langle s_i'-s_j', s_i'-s_j' \rangle = \dE(s_i',s_j')^2 = \dE(t_i', t_j')^2 = \langle t_i'-t_j', t_i'-t_j' \rangle.
    \end{align*}
    Hence, 
    \[
      \| s_i'\|_2+\|s_j'\|_2 -2\langle s_i',s_j'\rangle \ \ = \ \ \| t_i'\|_2+\|t_j'\|_2-2\langle t_i',t_j'\rangle.
    \]
    Using $\|s_i'\|_2 = \dE(0,s_i') = \dE(0,t_i')=\|t_i'\|_2$, we
    obtain that
    \begin{align*}
        \langle s_i', s_j' \rangle & \ \ = \ \ \langle t_i', t_j' \rangle.
    \end{align*}
    Let $\BB\coloneqq\{s_{i_1}',\ldots,s_{i_m}'\}$ be a basis of the vector space
    \[
       V \ \coloneqq \ span(s_1',\ldots,s_{n-1}'),
    \]
    let $W\coloneqq
    span(t_{i_1}',\ldots,t_{i_m}')$, and consider the linear
    function 
     \  $F : V \longrightarrow W$ \ defined via \ $F(s'_{i_j})\deff
     t'_{i_j}$ for each $j\in[1,m]$. \
    For this function $F$, the following holds for all $a,b\in[1,m]$:
    \begin{align*}
        \langle F(s_{i_a}'), F(s_{i_b}') \rangle \  \ = \ \ \langle
        t'_{i_a}, t'_{i_b}\rangle \ \ = \ \ \langle s'_{i_a}, s'_{i_b} \rangle.
    \end{align*}
    Thus, $F$ is orthogonal and the family
    $\{t_{i_1}',\ldots,t_{i_m}'\}$ is a basis of
    $span(t_1',\ldots,t_{n-1}')$.

    To see that $F(s'_i)=t'_i$ holds for all $i\in[0,n)$ consider the gramian matrix
    \begin{align*}
        G\deff&\ \left( \langle s'_i,s'_j \rangle \right)_{i,j\in[0,n)} \\
        =&\ \left( \langle t'_i,t'_j \rangle \right)_{i,j\in[0,n)}.
    \end{align*}
    Note that the gramian matrix $G$ is nonsingalur because $\mathfrak B$ is a basis.
    Now, for $i\in[0,n)$ and $a=(a_1,\cdots,a_m),b=(b_1,\cdots,b_m)\in\mathbb R^m$ such that
    \begin{align*}
        s'_i =&\ a_1 s'_{i_1}+ \cdots +a_m s'_{i_m} \text{ and} \\
        t'_i =&\ b_1 t'_{i_1}+ \cdots +b_m t'_{i_m}
    \end{align*}
    we have
    \begin{align*}
        G\cdot a = \sum_{h=1}^m a_h\langle s'_{i_h}, s'_{i_j} \rangle
        = \left\langle \sum_{h=1}^m a_h s'_{i_h}, s'_{i_j} \right\rangle
        = \langle s'_i,s'_j \rangle 
        = \langle t'_i,t'_j \rangle
        = \sum_{h=1}^m b_h\langle t'_{i_h},t'_{i_j} \rangle 
        = G\cdot b
    \end{align*}
    and thus $a=b$.
    Hence,
    \begin{align*}
        F(s'_i) = \sum_{h=1}^m a_h F(s'_{i_h})
        = \sum_{h=1}^m a_h t'_{i_h}
        = \sum_{h=1}^m b_h t'_{i_h}
        = t'_i
    \end{align*}
    i.\,e., $F(s'_i)=t'_i$ holds for all $i\in[0,n)$.

    From linear algebra we know that $F$ can be extended to an orthogonal function $F':\mathbb R^k\longrightarrow\mathbb R^k$.
    Consider the matrix $M'\in\orth(k)$ such that 
    \[
       F'(v)=w  \ \ \iff \ \ M' v=w
    \]
    is true for all $v\in\RR^k$. In particular, 
    $M's'_i=t'_i$ is true for all $i\in[0,n)$, and thus $M'\cdot S'=T'$.
    Hence,
    \[
        T-t_0 \ \ = \ \ T' \ \ = \ \ M'\cdot S' \ \ = \ \ M' \cdot (
        S-s_0) \ \ = \ \ M'\cdot S - s_0,
    \]
    and therefore \ $T = M'\cdot S + (t_0-s_0)$.
\end{proof}

\bigskip

\begin{proof}[of Lemma~\ref{lem:minwindow}] \ \\
    Let $S,T\in\mathcal T$ with $\#S = m\leq n = \#T$ and choose $i\in[0,n{-}m]$ such that $d_p(S,T_i^m)=d_p(S,T)$.
    Then, the desired inequality follows:
    \begin{align*}
        d_{\Norm{\cdot}}^\Delta(S,T) & \ \ \leq \ \ d_{\Norm{\cdot}}^\Delta(S,T_i^m) \\
        & \ \ \leq \ \  C(m)\cdot d_p(S,T_i^m) \\
        & \ \ = \ \ C(m)\cdot d_p(S,T)
    \end{align*}
\end{proof}

\bigskip

\begin{proof}[of Lemma~\ref{lem:quadmaxest}] \ \\
    First consider the case where $\#S=\#T = n$.
    Choose $j^*$ such that the $j^*$-th column of $\Delta S -\Delta T$
    has the maximum sum, i.\,e.,
    \begin{align*}
        \sum_{i=0}^{n-1-j^*} \left| \dE_2(s_i,s_{i+j^*}) -
          \dE_2(t_i,t_{i+j^*}) \right| \ \ = \ \ d^\Delta_m(S,T).
    \end{align*}
    Then, for $\dE\deff\dE_2$ we have
    \begin{align*}
        d^\Delta_m(S,T) &= \sum_{i=0}^{n-1-j^*} \left| \dE(s_i,s_{i+j^*}) - \dE(t_i,t_{i+j^*}) \right| \\
        & \leq \sum_{i=0}^{n-1-j^*} \big( \dE(s_i,t_i) + \dE(s_{i+j^*}, t_{i+j^*}) \big) \\
        & \leq 2\cdot \sum_{i=0}^{n-1} \dE(s_i,t_i) \\
        & = 2\cdot d_1(S,T)
    \end{align*}

    For $S,T\in\mathcal T$ of different lengths, the inequality follows using Lemma~\ref{lem:minwindow} and $C(n)\coloneqq 2$.
\end{proof}

\bigskip

\begin{proof}[of Lemma~\ref{lem:quadsumest}] \ \\
    First, assume $\#S=\#T = n$.
    Using $\|u+v\|_1=\|u\|_1+\|v\|$ for vectors $u,v\in\mathbb R^k$ with non-negative entries only we get the inequality:
    \begin{align*}
        d^\Delta_1(S,T) &\leq \left\| \big( d(s_i,t_i) + d(s_{i+j}, t_{i+j}) \big)_{0\leq i < n-1, 1\leq j<n-i} \right\|_1 \\
        &= \left\| ( n - i - 1)\cdot\big( d(s_i,t_i) \big)_{0\leq i<n-1} \right\|_1 + \\
        &\quad +\left\| \big( d(s_{i+j}, t_{i+j}) \big)_{0\leq i<n-1, 1\leq j<n-i} \right\|_1 \\
        &= \left\| \big( (n-i-1)\cdot d(s_i,t_i) \big)_{0\leq i<n-1} \right\|_1 + \\
        &\quad + \left\| \big( i\cdot d(s_i,t_i) \big)_{0<i\leq n-1} \right\|_1 \\
        &= \left\| \big( (n-1)\cdot d(s_i,t_i) \big)_{0\leq i < n} \right\|_1 \\
        &= (n-1)\cdot d_1(S,T)
    \end{align*}
    For the first inequation we used the triangle inequality and symmetry of $d$:
    \begin{align*}
        d(s_i,s_{i+j}) \leq d(s_i,t_i)+d(t_i,t_{i+j})+d(t_{i+j},s_{i+j})
    \end{align*}
    \begin{align*}
        \Longrightarrow |d(s_i,s_{i+j})-d(t_i,t_{i+j})| \leq d(s_i,t_i)+d(s_{i+j},t_{i+j})
    \end{align*}

    Considering $S,T\in\mathcal T$ with $\#S = m > n = \#T$, the inequality follows using Lemma~\ref{lem:minwindow} and $C(n) \coloneqq n-1$.
\end{proof}

\bigskip

\begin{claim}
    \label{cl:noupperbound}
    Consider $\MM=\RR^2$, $\dE=\dE_2$, and let $\epsilon>0$, $a\deff\sqrt{1-\epsilon^2}$.
    Then for $S,T_\epsilon\in\TS_3$ with
    \begin{align*}
        S\deff \left( \colvec{2}{-a}{0}, \colvec{2}{0}{0}, \colvec{2}{a}{0} \right) \text{ and }\ T_\epsilon\deff \left( \colvec{2}{-a}{0}, \colvec{2}{0}{\epsilon}, \colvec{2}{a}{0} \right)
    \end{align*}
    the following inequality holds:
    \begin{align*}
        \delta^C(S,T_\epsilon)\geq\frac \epsilon 2.
    \end{align*}
\end{claim}
\begin{proof}\ \\
    \begin{figure}
        \centering
        \includegraphics{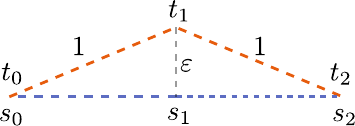}
        \caption{Two time series $S$ and $T_\epsilon$ in $\mathbb R^2$, such that $\frac{d_1^C(S,T_\epsilon)}{d_1^\Delta(S,T_\epsilon)}\xrightarrow{\varepsilon\rightarrow 0}+\infty$.}
        \label{fig:noupperbound}
    \end{figure}
    Denote $S=(s_0,s_1,s_2)$ and $T_\epsilon=(t_0,t_1,t_2)$.

    \emph{Assume} that $M\in\orth(k)$ and $v_0\in\MM$ exist such that
    \begin{align*}
        \delta(S,M\cdot T_\epsilon+v_0) < \frac \epsilon 2.
    \end{align*}
    Then, considering the linear function
    \begin{align*}
        \iota: \MM \longrightarrow&\MM \\
        v \longmapsto& M\cdot v+v_0,
    \end{align*}
    the following inequality must hold for each $i\in[0,2]$:
    \begin{align*}
        \dE(s_i,\iota(t_i)) < \frac \epsilon 2.
    \end{align*}
    Since $\iota$ is an isometric function and $\frac 1 2 s_0+\frac 1 2 s_2-s_1=0$, we have
    \begin{align*}
        \dE\left(s_1, \frac 1 2 \iota(t_0) + \frac 1 2 \iota(t_2)\right) =&\ \Norm{\frac 1 2 \iota(t_0) + \frac 1 2 \iota(t_2) - s_1}_2 \\
        =&\ \Norm{\frac 1 2 \iota(t_0) - \frac 1 2 s_0 + \frac 1 2 \iota(t_2) - \frac 1 2 s_2 + \frac 1 2 s_0 + \frac 1 2 s_2 - s_1}_2 \\
        \leq&\ \frac 1 2 \dE(\iota(t_0),s_0) + \frac 1 2 \dE(\iota(t_2),s_2) \\
        <&\ \frac \epsilon 2,
    \end{align*}
    i.\,e., we obtain a contradiction with
    \begin{align*}
        \dE(s_1, \iota(t_1)) \geq \underbrace{\dE\left(\iota(t_1), \iota\left( \frac 1 2 (t_0+t_2) \right)\right)}_{=\dE\left( t_1,\frac 1 2(t_0+t_2) \right)=\epsilon} - \dE\left( s_1,\iota\left( \frac 1 2 (t_0+t_2) \right) \right) > \epsilon - \frac\epsilon 2 > \frac\epsilon 2.
    \end{align*}
    Hence, $\delta^C(S,T_\epsilon) = \inf_{M\in\orth(k),v_0\in\MM} \delta(S,M\cdot T_\epsilon+v_0) \geq \frac \epsilon 2$.
\end{proof}

\bigskip

\section{Details Omitted in Section~5}

\begin{proof}[of Theorem~\ref{thm:qlinsumest}] \ \\
    The fiddly part of the calculation are the following inequalities for constant $k\in\mathbb N$:
    \begin{align*}
        \left|\left\{ j\in[1,k) \mid j\in 2^{\mathbb N} \right\} \right| &= \left| \left\{ 2^0, 2^1,\cdots,2^{\lfloor\log(k-1)\rfloor} \right\} \right| \\
        &= 1 + \lfloor \log(k-1) \rfloor
    \end{align*}
    For the second inequation in \eqref{eq:qlinsumest}, the numbers $i_l$ have to be choosen properly.
    \begin{align*}
        \left|\left\{ i+j=k\mid j\in 2^{\mathbb N}, i\geq 0 \right\} \right| &= \left| \bigcup_{l=0}^{\lfloor\log k\rfloor}\left\{ i_l+2^l \right\} \right| \\
        & =1+\lfloor\log k\rfloor
    \end{align*}

    First, consider $S,T\in\mathcal T$ of the same length $\#S=n=\#T$.
    Then, the desired inequality follows with the following calculation and subsequently applying Lemma~\ref{lem:minwindow}:
    \begin{align}
        & d^\delta_1(S,T) \nonumber \\
        &= \|\delta S-\delta T\|_1 \nonumber \\
        & =\left\|\left( \left| d(s_i,s_{i+j})-d(t_i,t_{i+j})\right| \right)_{i\in[0,n-1),j\in[1,n-i),j\in 2^{\mathbb N}} \right\| \nonumber \\
        & \leq \left\| \left(d(s_i,t_i) \right)_{i\in[0,n-1),j\in[1,n-i),j\in 2^{\mathbb N}}\right\| \nonumber \\
        & \quad + \left\| \left( d(s_{i+j},t_{i+j}) \right)_{i\in[0,n-1),j\in[1,n-i),j\in 2^{\mathbb N}} \right\| \nonumber \\
        & = \left\| \left( \lfloor\log(n-i-1)\rfloor\cdot d(s_i,t_i) \right)_{i\in[0,n-1)} \right\| + \nonumber \\
        & \quad + \left\| \left( (1+\lfloor\log k\rfloor )\cdot d(s_k,t_k) \right) _{k\in[1,n)} \right\| \nonumber \\
        & = \lfloor\log(n-1)\rfloor\cdot d(s_0,t_0) + \nonumber \\
        & \quad \left\| \left( \left( 1+\lfloor\log i\rfloor + \lfloor \log(n-i-1)\rfloor \right)\cdot d\left( s_i,t_i \right) \right)_{i\in[1,n-1)} \right\| \nonumber \\
        & \quad + \left( 1+\lfloor\log(n-1)\rfloor \right)\cdot d(s_{n-1},t_{n-1}) \nonumber \\
        & \leq \lfloor 2\cdot\log(n-1)\rfloor\cdot d_1(S,T)
        \label{eq:qlinsumest}
    \end{align}
    Using Lemma~\ref{lem:minwindow} we can allow $S,T$ to be of different lengths.
    Since the set of entries in $\delta T$ is a subset of the entries in $\Delta T$, the desired inequality follows analogously as in the proof of Theorem~\ref{thm:delta} by minimizing $d_1(S,M\cdot T+v_0)$.
\end{proof}